\def\UseBibLatex{1}
\ifx\SoCG\undefined%
\documentclass[11pt]{article}%
\providecommand{\SoCGVer}[1]{}%
\providecommand{\NotSoCGVer}[1]{#1}%
\else%
\makeatletter
\def\input@path{{lipics/}{../lipics/}}
\makeatother
\documentclass[a4paper,USenglish,cleveref,autoref,thm-restate]%
{socg-lipics-v2021}

\hideLIPIcs 
\providecommand{\SoCGVer}[1]{#1}%
\providecommand{\NotSoCGVer}[1]{}%
\fi

\IfFileExists{sariel_computer.sty}{\def\sarielComp{1}}{}

\ifx\sarielComp\undefined%
\newcommand{\SarielComp}[1]{}
\newcommand{\NotSarielComp}[1]{#1}%
\else
\newcommand{\SarielComp}[1]{#1}%
\newcommand{\NotSarielComp}[1]{}%
\fi
\newcommand{\IfPrinterVer}[2]{#2}%

\NotSoCGVer{%
\usepackage[cm]{fullpage}%
}

\usepackage{amsmath}%
\usepackage{amssymb}%
\usepackage[cmyk]{xcolor}%

\NotSoCGVer{%
   \usepackage{euscript}%
}
\NotSoCGVer{%
   \usepackage[amsmath,thmmarks]{ntheorem}%
   \theoremseparator{.}%
}

\NotSoCGVer{%
   \usepackage{titlesec}%
   \titlelabel{\thetitle. }%
   
   \titleformat{\paragraph}[runin]
   {\normalfont\bfseries}
   {\theparagraph}
   {1em}
   {\addperiod}
   
   \newcommand{\addperiod}[1]{#1.}%
}  

\usepackage{graphicx}%
\usepackage{xcolor}%
\usepackage{mleftright}%
\usepackage{xspace}%
\usepackage{hyperref}%
\usepackage{bm}

\usepackage{caption}%

\SarielComp{\usepackage{sariel_colors}}%

\IfPrinterVer{%
   \usepackage{hyperref}%
}{%
   \usepackage{hyperref}%
   \hypersetup{%
      breaklinks,%
      colorlinks=true,%
      urlcolor=[rgb]{0.25,0.0,0.0},%
      linkcolor=[rgb]{0.5,0.0,0.0},%
      citecolor=[rgb]{0,0.2,0.445},%
      filecolor=[rgb]{0,0,0.4},
      anchorcolor=[rgb]={0.0,0.1,0.2}%
   }
}

\providecommand{\BibLatexMode}[1]{}
\providecommand{\BibTexMode}[1]{#1}

\ifx\UseBibLatex\undefined%
  \renewcommand{\BibLatexMode}[1]{}
  \renewcommand{\BibTexMode}[1]{#1}
\else
  \renewcommand{\BibLatexMode}[1]{#1}
  \renewcommand{\BibTexMode}[1]{}
\fi

\newcommand{\UsePackage}[1]{%
  \IfFileExists{../styles/#1.sty}{%
      \usepackage{../styles/#1}%
   }{%
      \IfFileExists{./styles/#1.sty}{%
         \usepackage{styles/#1}%
      }{%
         \usepackage{#1}%
      }%
   }%
}

\BibLatexMode{%
   \UsePackage{sariel_biblatex}%
   \usepackage[english]{babel}%
   \usepackage{csquotes}
   \renewcommand*{\multicitedelim}{\addcomma\space}%
}

\SoCGVer{%
   \theoremstyle{plain}%
   \newtheorem{fact}[theorem]{Fact}
   \newtheorem{invariant}[theorem]{Invariant}
   \newtheorem{question}[theorem]{Question}
   \newtheorem{prop}[theorem]{Proposition}
   \newtheorem{openproblem}[theorem]{Open Problem}

   \theoremstyle{plain}%
   \newtheorem{defn}[theorem]{Definition}
   \newtheorem{problem}[theorem]{Problem}
   \newtheorem{xca}[theorem]{Exercise}
   \newtheorem{exercise_h}[theorem]{Exercise}
   \newtheorem{assumption}[theorem]{Assumption}%

   \newtheorem{proofof}{Proof of\!}%
}%
\NotSoCGVer{%
\theoremseparator{.}%

\theoremstyle{plain}%
\newtheorem{theorem}{Theorem}[section]

\newtheorem{lemma}[theorem]{Lemma}

\newtheorem{corollary}[theorem]{Corollary}
\newtheorem{claim}[theorem]{Claim}%

\theoremstyle{plain}%
\theoremheaderfont{\sf} \theorembodyfont{\upshape}%
\newtheorem*{remark:unnumbered}[theorem]{Remark}%
\newtheorem{remark}[theorem]{Remark}%
\newtheorem{definition}[theorem]{Definition}
\newtheorem{defn}[theorem]{Definition}

\newcommand{\myqedsymbol}{\rule{2mm}{2mm}}

\theoremheaderfont{\em}%
\theorembodyfont{\upshape}%
\theoremstyle{nonumberplain}%
\theoremseparator{}%
\theoremsymbol{\myqedsymbol}%
\newtheorem{proof}{Proof:}%

}

\definecolor{nalmostblack}{rgb}{0, 0, 0.7}
\providecommand{\emphic}[2]{%
   \textcolor{nalmostblack}{%
      \textbf{\emph{#1}}}%
   \index{#2}}

\providecommand{\emphi}[1]{\emphic{#1}{#1}}

\definecolor{almostblack}{rgb}{0, 0, 0.5}
\providecommand{\emphw}[1]{{\emph{{\textcolor{almostblack}{#1}}}}}%

\providecommand{\emphOnly}[1]{\emph{\textcolor{almostblack}{\textbf{#1}}}}

\numberwithin{figure}{section}%
\numberwithin{table}{section}%
\numberwithin{equation}{section}%

\newcommand{\SarielThanks}[1]{\thanks{Department of Computer Science;
      University of Illinois; 201 N. Goodwin Avenue; Urbana, IL,
      61801, USA; %
      \href{mailto:sariel.spam@illinois.edu}%
      {sariel@illinois.edu}; %
      \url{http://sarielhp.org/}. #1}}

\newcommand{\HLink}[2]{\hyperref[#2]{#1~\ref*{#2}}}
\newcommand{\HLinkSuffix}[3]{\hyperref[#2]{#1\ref*{#2}{#3}}}

\newcommand{\tablab}[1]{\label{table:#1}}
\newcommand{\tabref}[1]{\HLink{Table}{table:#1}}

\newcommand{\figlab}[1]{\label{fig:#1}}
\newcommand{\figref}[1]{\HLink{Figure}{fig:#1}}

\newcommand{\thmlab}[1]{{\label{theo:#1}}}
\newcommand{\thmref}[1]{\HLink{Theorem}{theo:#1}}

\newcommand{\corlab}[1]{\label{cor:#1}}
\newcommand{\corref}[1]{\HLink{Corollary}{cor:#1}}%

\providecommand{\deflab}[1]{\label{def:#1}}
\newcommand{\defref}[1]{\HLink{Definition}{def:#1}}

\newcommand{\clmlab}[1]{\label{claim:#1}}
\newcommand{\clmref}[1]{\HLink{Claim}{claim:#1}}

\newcommand{\seclab}[1]{\label{sec:#1}}
\newcommand{\secref}[1]{\HLink{Section}{sec:#1}}
\newcommand{\rectA}{\Mh{B}}%
\newcommand{\rectB}{\Mh{D}}%

\newcommand{\DW}{\times}
\newcommand{\shrinkDY}[2]{#1_{\boxminus #2}}
\newcommand{\Rects}{\Mh{\mathcal{R}}}%

\newcommand{\itemlab}[1]{\label{item:#1}}
\newcommand{\itemref}[1]{\HLinkSuffix{}{item:#1}{}}

\newcommand{\apndlab}[1]{\label{apnd:#1}}
\newcommand{\apndref}[1]{\HLink{Appendix}{apnd:#1}}

\newcommand{\remlab}[1]{\label{rem:#1}}
\newcommand{\remref}[1]{\HLink{Remark}{rem:#1}}%

\newcommand{\lemlab}[1]{\label{lemma:#1}}
\newcommand{\lemref}[1]{\HLink{Lemma}{lemma:#1}}%

\providecommand{\eqlab}[1]{}%
\renewcommand{\eqlab}[1]{\label{equation:#1}}

\providecommand{\remove}[1]{}%
\newcommand{\Set}[2]{\left\{ #1 \;\middle\vert\; #2 \right\}}
\newcommand{\pth}[2][\!]{\mleft({#2}\mright)}%

\newcommand{\ceil}[1]{\left\lceil {#1} \right\rceil}
\newcommand{\floor}[1]{\left\lfloor {#1} \right\rfloor}

\newcommand{\brc}[1]{\left\{ {#1} \right\}}
\newcommand{\cardin}[1]{\left| {#1} \right|}%

\renewcommand{\th}{th\xspace}

\renewcommand{\Re}{\mathbb{R}}%
\usepackage[inline]{enumitem}

\newlist{compactenumA}{enumerate}{5}%
\setlist[compactenumA]{topsep=0pt,itemsep=-1ex,partopsep=1ex,parsep=1ex,%
   label=(\Alph*)}%

\newlist{compactenuma}{enumerate}{5}%
\setlist[compactenuma]{topsep=0pt,itemsep=-1ex,partopsep=1ex,parsep=1ex,%
   label=(\alph*)}%

\newlist{compactenumI}{enumerate}{5}%
\setlist[compactenumI]{topsep=0pt,itemsep=-1ex,partopsep=1ex,parsep=1ex,%
   label=(\Roman*)}%

\newlist{compactenumi}{enumerate}{5}%
\setlist[compactenumi]{topsep=0pt,itemsep=-1ex,partopsep=1ex,parsep=1ex,%
   label=(\roman*)}%

\newlist{compactitem}{itemize}{5}%
\setlist[compactitem]{label=\ensuremath{\bullet}}%
\setlist[compactitem]{topsep=0pt,itemsep=-1ex,partopsep=1ex,parsep=1ex,%
   label=\ensuremath{\bullet}}%

\usepackage{stmaryrd}%
\providecommand{\IntRange}[1]{\mleft\llbracket #1 \mright\rrbracket}
\newcommand{\IRX}[1]{\IntRange{#1}}%

\usepackage{wasysym}

\newcommand{\disk}{\Mh{\ocircle}}
\newcommand{\diskVY}[2]{\disk_{\downarrow}^{#1}\pth{#2}}%
\newcommand{\si}[1]{#1}

\providecommand{\Mh}[1]{#1}%

\newcommand{\eps}{\varepsilon}

\newcommand{\FF}{\Mh{\mathcal{F}}}%

\newcommand{\DT}{\Mh{\mathcal{D}}}%

\newcommand{\DG}{\Mh{\mathcal{D}}}%

\newcommand{\etal}{\textit{et~al.}\xspace}

\newcommand{\Term}[1]{\textsf{#1}}

\newcommand{\QSPD}{\Term{QSPD}\xspace}

\newcommand{\StavThanks}[1]{%
   \thanks{Department of Computer Science; University of Illinois; 201
      N. Goodwin Avenue; Urbana, IL, 61801, USA; %
      \href{mailto:stava2.spam@illinois.edu}%
      {stava2@illinois.edu}; %
      \url{https://publish.illinois.edu/stav-ashur}. %
      #1}}

\newcommand{\pa}{\Mh{p}}%
\newcommand{\pb}{\Mh{q}}%
\newcommand{\pc}{\Mh{u}}%
\newcommand{\pd}{\Mh{v}}%

\newcommand{\px}{\Mh{x}}%
\newcommand{\py}{\Mh{y}}%
\newcommand{\pz}{\Mh{z}}%

\newcommand{\dGY}[2]{\Mh{\mathsf{d}}\pth{#1,#2}}%
\newcommand{\dGZ}[3]{\Mh{\mathsf{d}_{#1}}\pth{#2,#3}}%
\newcommand{\dY}[2]{\left\| #1  #2 \right\|}%

\newcommand{\angleX}[1]{\sphericalangle #1}

\newcommand{\dsY}[2]{\mathsf{d}\pth{#1,#2}}

\newcommand{\body}{\Mh{C}}%

\newcommand{\grid}{\Mh{\mathsf{K}}}%

\newcommand{\EucG}{\Mh{\EuScript{K}}_\PS}

\providecommand{\G}{\Mh{G}}%
\renewcommand{\G}{\Mh{G}}%
\newcommand{\GA}{\Mh{H}}%
\newcommand{\GB}{\Mh{I}}%

\providecommand{\GB}{\Mh{I}}%
\renewcommand{\GB}{\Mh{I}}%

\newcommand{\PS}{\Mh{P}}%
\newcommand{\PSA}{\Mh{Q}}%

\newcommand{\PX}{\Mh{X}}%
\newcommand{\PY}{\Mh{Y}}%

\newcommand{\DotProd}[2]{\permut{{#1},{#2}}}
\newcommand{\permut}[1]{\left\langle {#1} \right\rangle}

\newcommand{\Line}{\Mh{\ell}}%

\newcommand{\PSup}{\Mh{P}_\uparrow}%
\newcommand{\PSdown}{\Mh{P}_\downarrow}%

\newcommand{\QS}{\Mh{\mathcal{Q}}}%
\newcommand{\liftX}[1]{\mathrm{lift}\pth{#1}}%

\newcommand{\rect}{{\Mh{R}}}%

\newcommand{\EG}{\Mh{E}}%
\newcommand{\EGX}[1]{\Mh{E}\pth{#1}}%
\newcommand{\region}{\Mh{\mathcalb{r}}}%
\newcommand{\gminus}{-}%
\newcommand{\interiorX}[1]{\mathrm{int}\pth{#1}}%
\newcommand{\restrictY}[2]{#1 \cap {#2}}

\newcommand{\cpX}[1]{\Mh{\mathrm{c{}p}}\pth{#1}}%
\newcommand{\diamX}[1]{\mathrm{diam}\pth{#1}}%

\newcommand{\spread}{\Mh{\Phi}}
\newcommand{\spreadX}[1]{\spread\pth{#1}}

\newcommand{\WS}{\Mh{\mathcal{W}}}%
\newcommand{\WeightX}[1]{\Mh{\omega} \pth{#1}}
\newcommand{\diameterX}[1]{\mathrm{d{}i{}am}\pth{#1}}

\newcommand{\SSPD}{\Term{SSPD}\xspace}%

\newcommand{\PSB}{\Mh{B}}%
\newcommand{\PSC}{\Mh{C}}%

\newcommand{\PSX}{\Mh{X}}%
\newcommand{\PSY}{\Mh{Y}}%

\newcommand{\WSPD}{\Term{WSPD}\xspace}%

\newcommand{\epsA}{\Mh{\vartheta}}%

\newcommand{\GY}[2]{\Mh{\mathcal{S}}\pth{#1, #2}}%
\newcommand{\cen}{\Mh{c}}%
\newcommand{\Pair}{\Mh{\Xi}}%

\newcommand{\QSup}{\QS_{\uparrow}}
\newcommand{\QSdown}{\QS_{\downarrow}}

\newcommand{\SaveContent}[2]{%
   \expandafter\newcommand{#1}{#2}%
}

\newcommand{\RestatementOf}[2]{
   \noindent%
   \textbf{Restatement of #1.}
   {\em #2{}}%
}

\DeclareFontFamily{U}{BOONDOX-calo}{\skewchar\font=45 }
\DeclareFontShape{U}{BOONDOX-calo}{m}{n}{<-> s*[1.05] BOONDOX-r-calo}{}
\DeclareFontShape{U}{BOONDOX-calo}{b}{n}{<-> s*[1.05] BOONDOX-b-calo}{}
\DeclareMathAlphabet{\mathcalb}{U}{BOONDOX-calo}{m}{n}
\SetMathAlphabet{\mathcalb}{bold}{U}{BOONDOX-calo}{b}{n}
\DeclareMathAlphabet{\mathbcalb}{U}{BOONDOX-calo}{b}{n}

\newcommand{\CHX}[1]{\mathsf{ch}\pth{#1}}%

\newcommand{\rinX}[1]{\Mh{r}_{\mathrm{in}}\pth{#1}}%
\newcommand{\routX}[1]{\Mh{R}_{\mathrm{out}}\pth{#1}}%
\newcommand{\arX}[1]{\Mh{\mathsf{a{}r}}\pth{#1}}%
\newcommand{\Elp}{\Mh{\mathcal{E}}}

\newcommand{\cell}{\Mh{\mathsf{C}}}%

\newcommand{\Of}{\Mh{\mathcal{O}}}%
\newcommand{\Oeps}{\Mh{\mathcal{O}_\eps}}%
\newcommand{\gConst}{\Mh{\tau}}%
\newcommand{\xSlabX}[1]{{\protect\overleftrightarrow{#1}}}
\newcommand{\ySlabX}[1]{\updownarrow\!{#1}}
\newcommand{\widthX}[1]{\Mh{\mathsf{wd}}\pth{#1}} \smallskip%

\SoCGVer{%
   \newcommand{\myparagraph}[1]{%
      \noindent%
      \textbf{#1.}
   }%
}
\NotSoCGVer{%
   \newcommand{\myparagraph}[1]{%
      \paragraph{#1}
   }%
}

\providecommand{\TPDF}[2]{\texorpdfstring{#1}{#2}}

\newcommand{\sqr}{\mathcalb{s}}%
\newcommand{\sqrA}{\mathcalb{t}}%
\newcommand{\seg}{s}%
\newcommand{\origin}{o}%

\newcommand{\constA}{c_1}
\newcommand{\constB}{c_2}
\newcommand{\constC}{c_3}
\newcommand{\constD}{c_4}

\newcommand{\diamC}{\Mh{\mathcalb{d}}}%

\newcommand{\cone}{\Mh{\mathcalb{c}}}%
\newcommand{\coneB}{c}%
\newcommand{\ConeSet}{\Mh{\EuScript{C}}}%
\newcommand{\dir}{\Mh{\mathsf{n}}}
\newcommand{\nnZ}[3]{\Mh{\mathsf{nn}}_{#1} \pth{#2, #3}}

\newcommand{\Trap}{\Mh{T}}%
\newcommand{\Traps}{\Mh{\mathcal{T}}}%

\newcommand{\CC}{\Mh{\mathcal{C}}}%
\newcommand{\Body}{\CC}

\newcommand{\senseX}[1]{\Mh{\mathrm{sen}} \pth{#1}}%
\newcommand{\Triangles}{\bm\Delta}%

\BibLatexMode{\bibliography{ft_spanner}}

\title{Local Spanners Revisited}

\NotSoCGVer{%
   \author{%
      Stav Ashur%
      \StavThanks{}%
      \and%
      Sariel Har-Peled%
      \SarielThanks{Work on this paper was partially supported by a
         NSF AF award CCF-1907400.  }%
   }%
}%

\SoCGVer{%
   \author{Stav Ashur}%
   {Department of Computer Science, University of Illinois, 201
      N. Goodwin Avenue, Urbana, IL 61801, USA}%
   {stava2@illinois.edu}%
   {https://orcid.org/0000-0003-0533-8978}%
   {}%
   \author{Sariel Har-Peled}%
   {Department of Computer Science, University of Illinois, 201
      N. Goodwin Avenue, Urbana, IL 61801, USA}%
   {sariel@illinois.edu}%
   {https://orcid.org/0000-0003-2638-9635}%
   {Work on this paper was partially supported by a NSF AF award
      CCF-1907400.}%
}%

\SoCGVer{%
	\authorrunning{S. Ashur and S. Har-Peled} %
	\Copyright{Stav Ashur and Sariel Har-Peled}%
	\ccsdesc[500]{Theory of computation~Computational geometry}%
	\keywords{Geometric graphs, Fault-tolerant spanners}
	\Copyright{Stav Ashur and Sariel Har-Peled}
}

\NotSoCGVer{\date{\today}}

\begin{document}

\maketitle

\begin{abstract}
    For a set of points $\PS \subseteq \Re^2$, and a family of regions
    $\FF$, a \emph{local $t$-spanner} of $\PS$, is a sparse graph $\G$
    over $\PS$, such that, for any region $\region \in \FF$, the
    subgraph restricted to $\region$, denoted by
    $\restrictY{\G}{\region} = \G_{\PS \cap \region}$, is a
    $t$-spanner for all the points of $\region \cap \PS$.

    We present algorithms for the construction of local spanners with
    respect to several families of regions, such as homothets of a
    convex region. Unfortunately, the number of edges in the resulting
    graph depends logarithmically on the spread of the input point
    set. We prove that this dependency can not be removed, thus
    settling an open problem raised by Abam and Borouny.  We also show
    improved constructions (with no dependency on the spread) of local
    spanners for fat triangles, and regular $k$-gons. In particular,
    this improves over the known construction for axis parallel
    squares.

    We also study a somewhat weaker notion of local spanner where one
    allows to shrink the region a ``bit''. Any spanner is a weak local
    spanner if the shrinking is proportional to the
    diameter. Surprisingly, we show a near linear size construction of
    a weak spanner for axis-parallel rectangles, where the shrinkage
    is \emph{multiplicative}.
\end{abstract}

\section{Introduction}

For a set $\PS$ of points in $\Re^d$, the \emphw{Euclidean graph}
$\EucG = \bigl(\PS, \binom{\PS}{2}\bigr)$ of $\PS$ is an undirected
graph.  Here, an edge $\pa \pb \in \EG$ is associated with the segment
$\pa\pb$, and its weight is the (Euclidean) length of the segment.
Let $\G=(\PS,\EG)$ and $\GB=(\PS,\EG')$ be two graphs over the same
set of vertices (usually $\GB$ is a subgraph of $\G$). Consider two
vertices $\pa,\pb \in \PS$, and parameter $t \geq 1$.  A path $\pi$
between $\pa$ and $\pb$ in $\GB$, is a \emphw{$t$-path}, if the length
of $\pi$ in $\GB$ is at most $t\cdot\dGZ{\G}{\pa}{\pb}$, where
$\dGZ{\G}{\pa}{\pb}$ is the length of the shortest path between $\pa$
and $\pb$ in $\G$.  The graph $\GB$ is a \emphw{$t$-spanner} of $\G$
if there is a $t$-path in $\GB$, for any $\pa,\pb\in \PS$.  Thus, for
a set of points $\PS\subseteq \Re^d$, a graph $\G$ over $\PS$ is a
\emphw{$t$-spanner} if it is a $t$-spanner of the euclidean graph
$\EucG$. There is a lot of work on building geometric spanners, see
\cite{ns-gsn-07} and references there in.

\paragraph*{Fault-tolerant spanners}

An \emphw{$\FF$-fault-tolerant spanner} for $\PS\subseteq \Re^d$, is a
graph $\G=(\PS, \EG)$, such that for any region $\region$ (i.e., the
``attack''), the graph $\G \gminus \region$ is a $t$-spanner of
$\EucG - \region$ (See \defref{def:residual:graph} for a formal
definition of this notation).  Here $\G \gminus \region$ denotes the
graph after one deletes from $\G$ all the vertices in
$\PS \cap \region$, and all the edges in $\G$ that their corresponding
segments intersect $\region$.  Surprisingly, as shown by Abam \etal
\cite{abfg-rftgs-09}, such fault-tolerant spanners can be constructed
where the attack region is any convex set. Furthermore, these spanners
have a near linear number of edges.

Fault-tolerant spanners were first studied with vertex and edge
faults, meaning that some arbitrary set of maximum size $k$ of
vertices and edges has failed. Levcopoulos \etal \cite{lns-iacfts-02}
showed the existence of $k$-vertex/edges fault tolerant spanners for a
set of points $\PS$ in some metric space. Their spanner had
$\Of(kn\log n)$ edges, and weight, i.e. sum of edge weights, bounded
by $f(k)\cdot wt(MST(\PS))$ for some function $f$. Lukovszki
\cite{l-nrftgs-99} later achieved a similar construction, improving
the number of edges to $\Of(kn)$, and was able to prove that the
result is asymptotically tight.

\paragraph*{Local spanners}

Recently, Abam and Borouny \cite{ab-lgs-21} introduced the notion of
local spanners.  For a family of regions $\FF$, a graph
$\G = (\PS,\EG)$ is a \emphw{local $t$-spanner} for $\FF$, if for any
$\region \in \FF$, the subgraph of $\G$ induced on $\PS \cap \region$
is a $t$-spanner.
Specifically, this induced subgraph $ \restrictY{\G}{\region}$
contains a $t$-path between any $\pa,\pb \in \PS \cap \region$ (note
that we keep an edge in the subgraph only if both its endpoints are in
$\region$, see \defref{def:residual:graph}).

Abam and Borouny \cite{ab-lgs-21} showed how to construct such
spanners for axis-parallel squares and vertical slabs. In this work,
we are further extend their results.  They also showed how to
construct such spanners for disks if one is allowed to add Steiner
points. Abam and Borouny left the question of how to construct local
spanners for disks as an open problem.

To appreciate the difficulty in constructing local spanners, observe
that unlike regular spanners, the construction has to take into
account many different scenarios as far as which points are available
to be used in the spanner. As a concrete example, a local spanner for
axis-parallel rectangle requires quadratic number of edges, see
\figref{bad:rectangles}.

\begin{figure}[h]
    \centerline{\includegraphics{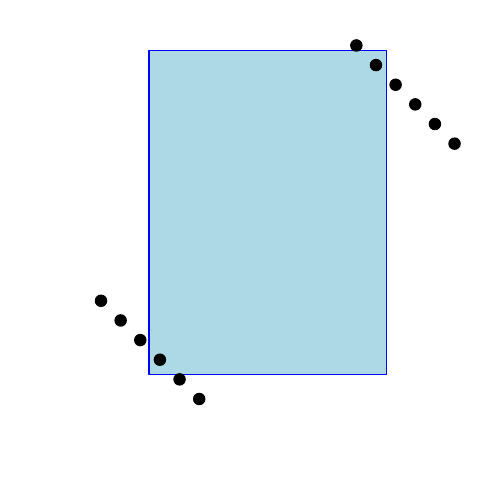}}
    \caption{For any point in the top diagonal and bottom diagonal,
       there is a fat axis parallel rectangle that contains only these
       two points. Thus, a local spanner requires quadratic size in
       this case.  }
    \figlab{bad:rectangles}
\end{figure}

Namely, regular spanners can rely on using midpoints in their path
under the assurance that they are always there. For local spanners
this is significantly harder as natural midpoints might
``disappear''. Intuitively, a local spanner construction needs to use
midpoints that are guaranteed to be present judging only from the
source and destination points of the path.

\paragraph*{A good jump is hard to find}

Most constructions for spanners can be viewed as searching for a way
to build a path from the source to the destination by finding a
``good'' jump, either by finding a way to move locally from the source
to a nearby point in the right direction, as done in the
$\theta$-graph construction, or alternatively, by finding an edge in
the spanner from the neighborhood of the source to the neighborhood of
the destination, as done in the spanner constructions using
well-separated pairs decomposition (\WSPD). Usually, one argues
inductively that the spanner must have (sufficiently short) paths from
the source to the start of the jump, and from the end of the jump to
the destination, and then, combining these implies that the resulting
new path is short.  These ideas guide our constructions as
well. However, the availability of specific edges depends on the query
region, making the search for a good jump significantly more
challenging. The constructions have to guarantee that there are many
edges available, and that at least one of them is useful as a jump
regardless of the chosen region.

\begin{table}[t]
    \centering

    \begin{tabular}{|l|c|c||c|c|}
      \hline
      Region
      &
        \# edges
      &
        Paper
      &
        New \# edges
      &
        Location in paper
      \\
      \hline
      \multicolumn{5}{c}{ Local $(1+\eps)$-spanners$\Bigr.$}
      \\
      \hline
      Halfplanes
      &
        $\Of( \eps^{-2} n \log n )\bigr.$
      &
        \cite{abfg-rftgs-09}
      &
      &
      \\
      \hline
      Axis-parallel squares
      &
        $\Oeps (n \log^6 n) \Bigr.$
      &
        \cite{ab-lgs-21}
      &
        $\Of\pth{\eps^{-3} n \log n}$
      &
        \remref{improved}%
      \\
      \hline
      Vertical slabs
      &
        $\Of (\eps^{-2} n \log n) \Bigr.$
      &
        \cite{ab-lgs-21}
      &
      &
      \\
      \hline
      Disks+Steiner points
      &
        $\Oeps (n ) \Bigr.$
      &
        \cite{ab-lgs-21}
      &
      &
      \\
      \hline
      Disks
      &
      &
      &
        $\Of\pth{ \eps^{-2} n\log \spread  }\Bigr.$
      &
        \thmref{main:1}%
      \\

      \cline{4-5}

      &
      &
      &
        $\Omega(n \log \spread)\Bigr.$
      &
        \lemref{l:s:lower:bound}%
      \\
      \hline
      Homothets convex shape
      &
      &
      &
        $\Of\pth{ \eps^{-2} n\log \spread  }\Bigr.$
      &
        \thmref{main:1}%
      \\
      \hline
      Homothets $\alpha$-fat triangles
      &
      &
      &
        $\Of\pth{ (\alpha\eps)^{-1} n}\Bigr.$
      &
        \thmref{l:s:triangle}%
      \\
      \hline
      Homothets triangles
      &
      &
      &       
        $\Omega\pth{ n \log \spread }\Bigr.$
      &
        \lemref{l:b:triangles}%
      \\
      \hline
      \multicolumn{5}{c}{$\delta$-weak local $(1+\eps)$-spanners$\Bigr.$}
      \\
      \hline
      Bounded convex shape
      &
      &
      &
        $\Of\bigl( (\eps^{-1} + \delta^{-2}) n \bigr)\Bigr. $
      &
        \lemref{w:l:s:regions}%
      \\
      \hline
      \multicolumn{5}{c}{$(1-\delta)$-local $(1+\eps)$-spanners$\Bigr.$}
      \\
      \hline%
      Axis-parallel rectangles
      &
      &
      &
        $\Of\bigl((\eps^{-2} + \delta^{-2}) n \log^2 n \bigr)\Bigr.$
      &
        \thmref{a:l:s:rectangles}%
      \\
      \hline

    \end{tabular}
    \smallskip%
    \caption{Known and new results. The notation $\Oeps$ hides
       polynomial dependency on $\eps$ which is not specified in the
       original work.}
    \tablab{results_summary}
\end{table}

\subsection*{Our results}

Our results are summarized in \tabref{results_summary}.

\paragraph*{Almost local spanners}

We start by showing that regular geometric spanners are local spanners
if one is required provide the spanner guarantee only to shrunken
regions. Namely, if $\G$ is a $(1+\eps)$-spanner of $\PS$, then for
any convex region $\Body$, the graph $ \restrictY{\G}{\Body}$ is a
spanner for $\Body' \cap \PS$, where $\Body'$ is the set of all points
in $\Body$ that are in distance at least $\eps\cdot \diamX{\Body}$
from its boundary.

\paragraph*{Homothets}
A \emphw{homothet} of a convex region $\Body$, is a translated and
scaled copy of $\Body$.  In \secref{disks} we present a construction
of spanners, which surprisingly, is not only fault-tolerant for all
convex regions, but is also a local spanner for homothets of a
prespecified convex region.  This in particular works for disks, and
resolves the aforementioned open problem of Abam and Borouny
\cite{ab-lgs-21}. Our construction is somewhat similar to the original
construction of Abam \etal \cite{abfg-rftgs-09}. For a parameter
$\eps>0$ the construction of a $(1+\eps)$-local spanner for homothets
takes $\Of\pth{ \eps^{-2} n\log \spread \log n }$ time, and the
resulted spanner is of size $\Of\pth{ \eps^{-2} n\log \spread }$,
where $\spread$ is the spread of the input point set $\PS$, and
$n=|\PS|$. We also provide a lower bound showing that this logarithmic
dependency on $\spread$ cannot be avoided.

The dependency on the spread $\spread$ in the above construction is
somewhat disappointing. However.  the lower bound constructions,
provided in \secref{lower:bound}, show that this is unavoidable for
disks or homothets of triangles.

Thus, the natural question is what are the cases where one can avoid
the ``curse of the spread'' -- that is, cases where one can construct
local spanners of near-linear size independent of the spread of the
input point set.

\paragraph*{The basic building block: $\Body$-Delaunay triangulation}

A key ingredient in the above construction is the concept of Delaunay
triangulations induced by homothets of a convex body. Intuitively, one
replaces the unit disk (of the standard $L_2$-norm) by the provided
convex region. It is well known \cite{cd-vdbcdf-85} that such diagrams
exist, have linear complexity in the plane, and can be computed
quickly.  In \secref{del:homothets} we review these results, and
restate the well-known property that the $\Body$-Delaunay
triangulation is connected when restricted to a homothet of $\Body$.
By computing these triangulations for carefully chosen subsets of the
input point set, we get the results stated above.

Specifically, we use well-separated and semi-separated decompositions
to compute these subsets.

\paragraph*{Fat triangles}
In \secref{triangles} we give a construction of local spanners for the
family $\FF$ of homothets of a given triangle $\triangle$, and get a
spanner of size $\Of\pth{ (\alpha\eps)^{-1}n }$ in
$\Of\pth{(\alpha\eps)^{-1}n\log n}$ time, where $\alpha$ is the
smallest angle in $\triangle$. This construction is a careful
adaptation of the $\theta$-graph spanner construction to the given
triangle, and it is significantly more technically challenging than
the original construction.

\paragraph*{$k$-regular polygons}

It seems natural that if one can handle fat triangles, then homothets
of $k$-regular polygons should readily follow by a simple
decomposition of the polygon into fat triangles. Maybe surprisingly,
this is not the case -- a critical configuration might involve two
points that are on the interior of two non-adjacent edges of a
homothet of the input polygon. We overcome this by first showing that
sufficiently narrow trapezoids, provide us with a good jump somewhere
inside the trapezoid, assuming one computes the Delaunay triangulation
induced by the trapezoid, and that the source and destination lie on
the two legs of the trapezoid. Next, we show that such a polygon can
be covered by a small number of narrow trapezoids and fat
triangles. By building appropriate graphs for each trapezoid/triangle
in the collection, we get a spanner for homothets of the given
$k$-regular polygon, with size that has no dependency on the
spread. Of course, the size does depend on $k$.  See \secref{k:gons}
for details, and \thmref{k:gon} for the precise result.

\paragraph*{Quadrant separated pair decomposition (\QSPD)}

In \apndref{qspd}, we describe a novel pair-decompos\-\si{ition}.
Specifically, the \QSPD breaks the input point set $\PS$ into pairs,
such that for any pair $\{\PX, \PY\}$ we have the property that there
is a translated axis system such that $\PX$ and $\PY$ belong to two
antipodal quadrants.  In $d$ dimensions there is such a decomposition
with $\Of( n \log^{d-1} n)$ pairs, and total weight
$\Of( n \log^{d} n)$.  A somewhat similar idea was used by Abam and
Borouny \cite{ab-lgs-21} for the $d=1$ case. We believe this
decomposition might be useful and is of independent interest.

\paragraph*{Multiplicative weak local spanner for rectangles}

In \apndref{w:l:s:rect}, we use \QSPD{}s to construct a weak local
spanner for axis parallel rectangles.  Here, the constructed graph
$\G$ over $\PS$, has the property that for any axis-parallel rectangle
$\rect$, the graph $\restrictY{\G}{\rect}$ is a $(1+\eps)$-spanner for
all the points of $\bigl((1-\eps)\rect \bigr)\cap \PS$, where
$(1-\eps)\rect$ is the scaling of the rectangle by a factor of
$1-\eps$ around its center. Importantly, this works for narrow
rectangles where this form of multiplicative shrinking is still
meaningful (unlike the diameter based shrinking mentioned
above). Contrast this with the lower bound (illustrated in
\figref{bad:rectangles}) of $\Omega(n^2)$ on the size of local spanner
if one does not shrink the rectangles. See \thmref{a:l:s:rectangles}
for details of the precise result.

\bigskip

See \tabref{results_summary} for a summary of known results and
comparisons to the results of this paper.

\section{Preliminaries}

\paragraph*{Residual graphs}
\begin{definition}
    \deflab{def:residual:graph} Let $\FF$ be a family of regions in
    the plane. For a fault region $\region \in \FF$ and a geometric
    graph $\G$ on a point set $\PS$, let $\G \gminus \region$ be the
    residual graph after removing from it all the points of $\PS$ in
    $\region$ and all the edges that their corresponding segments
    intersect $\region$. Similarly, let $\restrictY{\G}{\region}$
    denote the graph restricted to $\region$.  Formally, let
    \begin{equation*}
	\G \gminus \region%
	=%
	\bigl( \PS \setminus \region, \Set{ uv \in \EG }{ uv \cap
           \interiorX{\region} = \emptyset} \bigr)
	\qquad\text{and}\qquad%
	\restrictY{\G}{\region}%
	=%
	\bigl( \PS \cap {\region},
	\Set{uv \in \EG}{ uv \subseteq {\region} } \bigr).
    \end{equation*}
    where $\interiorX{\region}$ denotes the interior of $\region$.
\end{definition}

\subsection{On various pair decompositions}

For sets $\PSX, \PSY$, let
\begin{math}
    \PSX \otimes \PSY%
    =%
    \Set{\brc{\px,\py}}{ \px\in \PSX,\, \py\in \PSY, \px \ne \py }
\end{math}
be the set of all the (unordered) pairs of points formed by the sets
$\PSX$ and $\PSY$.

\begin{defn}[Pair decomposition]
    \deflab{pair:decomposition}%
    For a point set $\PS$, a \emphic{pair
       decomposition}{pair!decomposition} of $\PS$ is a set of pairs
    \begin{equation*}
        \WS = \brc{\bigl. \brc{\PSX_1,\PSY_1},\ldots,\brc{\PSX_s,\PSY_s}},
    \end{equation*}
    such that
    \begin{enumerate*}[label=(\Roman*)]
        \item $\PSX_i,\PSY_i\subseteq \PS$ for every $i$,
        \item $\PSX_i \cap \PSY_i = \emptyset$ for every $i$, and
        \item
        $\bigcup_{i=1}^s \PSX_i \otimes \PSY_i = \PS \otimes \PS$.
    \end{enumerate*}
    Its \emphi{weight} is
    $\WeightX{\WS} = \sum_{i=1}^s \pth{ \cardin{\PSX_i } +
       \cardin{\PSY_i}}$.
\end{defn}

The \emphi{closest pair} distance of a set of points
$\PS \subseteq \Re^d$, is
\begin{math}
    \cpX{\PS} = \underset{\pa, \pb \in \PS, \pa \neq \pb}{\min}
    \dY{\pa}{\pb}.
\end{math}
The \emphi{diameter} of $\PS$ is
\begin{math}
    \diamX{\PS} = \underset{\pa, \pb \in \PS}{\max}\dY{\pa}{\pb}.
\end{math}
The \emphi{spread} of $\PS$ is
$\spreadX{\PS} = \diamX{\PS} / \cpX{\PS}$, which is the ratio between
the diameter and closest pair distance.  While in general the weight
of a \WSPD (defined below) can be quadratic, if the spread is bounded,
the weight is near linear.  For $\PSX, \PSY \subseteq \Re^d$, let
$\dsY{\PSX}{\PSY} = \underset{\pa \in \PSX, \pb \in \PSY}{\min}
\dY{\pa}{\pb}$ be the \emphi{distance} between the two sets.

\begin{defn}
    \deflab{well:separated}%
    \deflab{WSPD}%
    Two sets $\PSX, \PSY \subseteq \Re^d$ are
    \begin{align*}
      \text{\emphi{$1/\eps$-well-separated}}
      \qquad
      &\text{if}\qquad
        \mathbf{max} \pth{ \diameterX{\PSX}, \diameterX{\PSY} } \leq
        \eps \cdot \dsY{\PSX}{\PSY},
      \\
      \text{and} \qquad \text{\emphi{$1/\eps$-semi-separated}}
      \qquad
      &\text{if}\qquad
        \mathbf{min} \pth{ \diameterX{\PSX}, \diameterX{\PSY} }
        \leq
        \eps \cdot \dsY{\PSX}{\PSY}.
    \end{align*}
    For a point set $\PS$, a \emphOnly{well-separated pair
       decomposition} (\emphOnly{W{S}{P}D{}}) of $\PS$ with parameter
    $1/\eps$ is a pair decomposition of $\PS$ with a set of pairs
    \begin{math}
        \WS = \brc{\bigl.
           \brc{\PSB_1,\PSC_1},\ldots,\brc{\PSB_s,\PSC_s}},
    \end{math}
    such that for all $i$, the sets $\PSB_i$ and $\PSC_i$ are
    $(1/\eps)$-separated. The notion of $1/\eps$-\SSPD (a.k.a
    \emphi{semi-separated pairs decomposition}) is defined
    analogously.
\end{defn}

\begin{lemma}[\cite{ah-ncsa-12}]
    \lemlab{s:s:p:d:spread}%
    Let $\PS$ be a set of $n$ points in $\Re^d$, with spread
    $\spread = \spreadX{\PS}$, and let $\eps > 0$ be a
    parameter. Then, one can compute a $(1/\eps)$-\WSPD $\WS$ for
    $\PS$ of total weight
    $\WeightX{\WS} = \Of(n \eps^{-d} \log \spread)$. Furthermore, any
    point of $\PS$ participates in at most
    $\Of\pth{ \eps^{-d} \log \spread}$ pairs.
\end{lemma}

\begin{theorem}[\cite{ah-ncsa-12,h-gaa-11}]
    \thmlab{S:S:P:D:main}%
    Let $\PS$ be a set of $n$ points in $\Re^d$, and let $\eps > 0$ be
    a parameter. Then, one can compute a $(1/\eps)$-\SSPD for $\PS$ of
    total weight $\Of\pth{n \eps^{-d} \log n}$. The number of pairs in
    the \SSPD is $\Of\pth{n \eps^{-d}}$, and the computation time is
    $\Of\pth{ n \eps^{-d} \log n }$.
\end{theorem}
 
\begin{lemma}
    \lemlab{chop:easy}%
    Given an $\alpha$-\SSPD $\WS$ of a set $\PS$ of $n$ points in
    $\Re^d$ and a parameter $\beta \geq 2$, one can refine $\WS$ into
    an $\alpha\beta$-\SSPD $\WS'$, such that
    $|\WS'| = \Of(|\WS|/\beta^d)$ and
    $\WeightX{\WS'} = \Of(\WeightX{\WS}/\beta^d)$.
\end{lemma}

\begin{proof}
    The algorithm scans the pairs of $\WS$. For each pair
    $\Pair = \{ \PSX, \PSY \} \in \WS$, assume that
    $\diameterX{\PSX} < \diameterX{\PSY}$. Let $\sqr$ be the smallest
    axis-parallel cube containing $\PSX$, and denote its sidelength by
    $r$.  Let $r' = r / \ceil{\sqrt{d} \beta }$.  Partition $\sqr$
    into a grid of cubes of sidelength $r'$, and let $T_\Pair$ be the
    resulting set of squares. The algorithm now add the set pairs
    \begin{equation*}
        \Set{ \{ \PSX \cap \sqrA, \PSY \} }{ \sqrA \in T_\Pair }
    \end{equation*}
    to the output \SSPD. Clearly, the resulting set is now
    $\alpha\beta$-semi separated, as we chopped the smaller part of
    each pair into $\beta$ smaller portions.
\end{proof}

\begin{defn}%
    \deflab{angular:separated}%
    An \emphi{$\eps$-double-wedge} is a region between two lines,
    where the angle between the two lines is at most $\eps$.

    Two point sets $\PX$ and $\PY$ that each lie in their own cone of
    a shared $\eps$-double-wedge are \emphi{$\eps$-angularly
       separated}.
\end{defn}

\SaveContent{\LemmaRefineDWBody}{%
   Given a $(1/\eps)$-\SSPD $\WS$ of $n$ points in the plane, one can
   refine $\WS$ into a $(1/\eps)$-\SSPD $\WS'$, such that each pair
   $\Pair = \{ \PSX, \PSY \} \in \WS'$ is contained in a
   $\eps$-double-wedge $\DW_\Pair$, such that $\PSX$ and $\PSY$ are
   contained in the two different faces of the double wedge
   $\DW_\Pair$. We have that $|\WS'| = \Of(|\WS|/\eps)$ and
   $\WeightX{\WS'} = \Of(\WeightX{\WS}/\eps)$. The construction time
   is proportional to the weight of $\WS'$.%
}

\begin{lemma}[Proof in \apndref{refine:d:w}]
    \lemlab{refine:d:w}%
    \LemmaRefineDWBody{}%
\end{lemma}

\begin{corollary}
    \corlab{S:S:P:D:angular}%
    Let $\PS$ be a set of $n$ points in the plane, and let $\eps > 0$
    be a parameter. Then, one can compute a $(1/\eps)$-\SSPD for $\PS$
    such that every pair is $\eps$-angularly separated.  The total
    weight of the \SSPD is $\Of\pth{n \eps^{-3} \log n}$, the number
    of pairs in the \SSPD is $\Of\pth{n \eps^{-3}}$, and the
    computation time is $\Of\pth{ n \eps^{-3} \log n }$.
\end{corollary}

\subsection{Weak local spanners for fat convex regions}
\seclab{convex regions}

\begin{defn}
    Given a convex region $\body$, let
    \begin{equation*}
        \shrinkDY{\body}{\delta}%
        =%
        \Set{ \pa \in \body }{ \dsY{\pa}{ \Re^2 \setminus \body} \geq \delta \cdot
           \diameterX{\body}}.
    \end{equation*}
    Formally, $\shrinkDY{\body}{\delta}$ is the Minkowski difference
    of $\body$ with a disk of radius $\delta \cdot \diameterX{\body}$.
\end{defn}

\begin{defn}
    Consider a (bounded) set $\body$ in the plane. Let $\rinX{\body}$
    be the radius of the largest disk contained inside $\body$.
    Similarly, $\routX{\body}$ is the smallest radius of a disk
    containing $\body$.

    The \emphi{aspect ratio} of a region $\body$ in the plane is
    $\arX{\body} = \routX{\body}/\rinX{\body}$. Given a family $\FF$
    or regions in the plane, its \emphw{aspect ratio} is
    $\arX{\FF} = \max_{\body \in \FF} \arX{\body}$.
\end{defn}

Note, that if a convex region $\body$ has bounded aspect ratio, then
$\shrinkDY{\body}{\delta}$ is similar to the result of scaling $\body$
by a factor of $1-\Of(\delta)$. On the other hand, if $\body$ is long
and skinny then this region is much smaller. Specifically, if $\body$
has width smaller than $2 \delta \cdot \diameterX{\body}$, then
$\shrinkDY{\body}{\delta}$ is empty.

\begin{lemma}
    \lemlab{w:l:s:regions}%
    Given a set $\PS$ of $n$ points in the plane, and parameters
    $\delta, \eps \in (0,1)$.  One can construct a graph $\G$ over
    $\PS$, in $\Of( (\eps^{-1} + \delta^{-2}) n \log n)$ time, and
    with $\Of\bigl( (\eps^{-1} + \delta^{-2}) n \bigr) $ edges, such
    that for any (bounded) convex $\body$ in the plane, we have that
    for any two points
    $\pa, \pb \in \PS \cap \shrinkDY{\body}{\delta}$ the graph
    $\restrictY{ \body }{\PS}$ has $(1+\eps)$-path between $\pa$ and
    $\pb$.
\end{lemma}

\begin{proof}
    Let $\epsA = \min( \eps, \delta^2)$. Construct, in
    $\Of(\epsA^{-1} n \log n)$ time, a standard $(1+\epsA)$-spanner
    $\G$ for $\PS$ using $\Of( \epsA^{-1} n)$ edges
    \cite{ams-dagss-99}.

    So, consider any body $\body \in \FF$, and any two vertices
    $\pa, \pb \in \PS \cap \body'$, where
    $\body'=\shrinkDY{\body}{\delta}$, Let $\ell = \dY{\pa}{\pb}$, let
    $\pi$ be the shortest path between $\pa$ and $\pb$ in $\G$, and
    let $\Elp$ be the locus of all points $\pc$, such that
    $\dY{\pa}{\pc} + \dY{\pc}{\pb} \leq (1+\epsA) \ell$. The region
    $\Elp$ is an ellipse that contains $\pi$. The furthest point from
    the segment $\pa\pb$ in this ellipse is realized by the co-vertex
    of the ellipse. Formally, it is one of the two intersection points
    of the boundary of the ellipse with the line orthogonal to
    $\pa \pb$ that passes through the middle point $\cen$ of this
    segment, see \figref{ellipse}. Let $\pz$ be one of these points.

    \begin{figure}[h]
        \centerline{\includegraphics{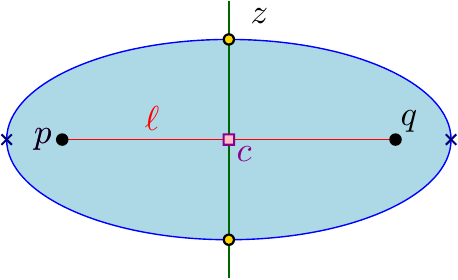}}
        \caption{An illustration of the settings in the proof of
           \lemref{w:l:s:regions} with $\Elp$ shown in blue.}
        \figlab{ellipse}
    \end{figure}

    We have that
    \begin{math}
        \dY{\pa}{\pz} = (1+\epsA)\ell/2.
    \end{math}
    Setting $h = \dY{\pz}{\cen}$, we have that
    \begin{equation*}
        h%
        =%
        \sqrt{\dY{\pa}{\pz}^2 - \dY{\pa}{\cen}^2 }
        =%
        \frac{\ell}{2} \sqrt{(1+\epsA)^2 - 1}
        =%
        \frac{\sqrt{\epsA (2+\epsA)}}{2} \ell%
        \leq
        \sqrt{\epsA} \ell
        \leq
        \sqrt{\epsA}\cdot \diameterX{\body}.
    \end{equation*}
    as $\ell \leq \diameterX{\body'} \leq \diameterX{\body}$.

    For any point $\px \in \body'$, we have that
    $\dsY{\px}{\Re^2 \setminus \body} \geq \delta \cdot
    \diamX{\body}$.  As such, to ensure that
    $\pi \subseteq \Elp \subseteq \body$, we need that
    \begin{math}
        \delta \cdot \diamX{\body} \geq h,
    \end{math}
    which holds if
    \begin{math}
        \delta \cdot \diamX{\body} \geq \sqrt{\epsA} \cdot
        \diameterX{\body}.
    \end{math}
    This in turn holds if $\epsA \leq \delta^2$. Namely, we have the
    desired properties if $\epsA = \min( \eps, \delta^2)$.
\end{proof}

\section{Local spanners of homothets of convex region}
\seclab{disks}

Let $\CC$ be a bounded convex and closed region in the plane (e.g., a
disk).  A \emphi{homothet} of $\CC$ is a scaled and translated copy of
$\CC$.  A point set $\PS$ is in \emphi{general position} with respect
to $\CC$, if no four points of $\PS$ lie on the boundary of a homothet
of $\CC$, and no three points are colinear.

A graph $\G=(\PS, \EG)$ is a \emphw{$\CC$-local $t$-spanner} for $\PS$
if for any homothet $\region$ of $\CC$ we have that
$\restrictY{\G}{\region}$ is a $t$-spanner of $\EucG \cap \region$.

\subsection{Delaunay triangulation for homothets}
\seclab{del:homothets}

\begin{defn}[\cite{cd-vdbcdf-85}]
    \deflab{c:del:triang}%
    Given $\CC$ as above, and a point set $\PS$ in general position
    with respect to $\CC$, the \emphi{$\CC$-Delaunay triangulation} of
    $\PS$, denoted by $\DG_{\CC}(\PS)$, is the graph formed by edges
    between any two points $\pa,\pb\in \PS$ such that there is a
    homothet of $\CC$ that contains only $\pa$ and $\pb$ and no other
    point of $\PS$.
\end{defn}

\begin{theorem}[\cite{cd-vdbcdf-85}]
    \thmlab{c:t:construct_delaunay_graph}%
    For any convex shape $\CC$ and a set of points $\PS$,
    $\DG_{\CC}(\PS)$ can be computed in $\Of(n \log n)$ time.
    Furthermore, the triangulation $\DG_{\CC}(\PS)$ has $\Of(n)$
    edges, vertices, and faces.
\end{theorem}

\begin{figure}[h]
    \includegraphics[page=1]{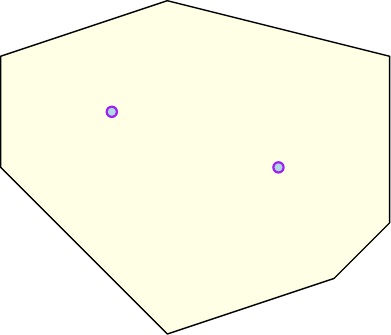}%
    \hfill%
    \includegraphics[page=2]{figs/shrink_2}%
    \hfill%
    \includegraphics[page=3]{figs/shrink_2}
    \caption{Shrinking of homothets so two points becomes on the
       boundary of the homothet.}
    \figlab{shrink:h}
\end{figure}

\begin{lemma}
    \lemlab{shrink:shrank}%
    Let $\CC$ be a convex bounded body, and let $\PS$ be a set of
    points in general position with respect to $\CC$. Then, if $C$ is
    a homothet of $\CC$ that contains two points
    $\pa, \pb \in \CC \cap \PS$, then there exists a homothet
    $C' \subseteq C$ of $\CC$ such that $\pa, \pb \in \partial C'$.
\end{lemma}
\begin{proof}
    The idea is to apply a shrinking process of $C$, as illustrated in
    \figref{shrink:h}.  Consider the mapping
    $f_{\beta, \pd} : \px \mapsto \beta (\px - \pd) + \pd $. It is a
    scaling of the plane around $\pd$ by a factor of $\beta$. Let
    $\beta'$ be the minimum value of $\beta$ such that
    $C_1 = f_{\beta,\pa}(C)$ contains $\pb$ (i.e., we shrink $C$
    around $\pa$ till $\pb$ becomes a boundary point). Next, shrink
    $C'$ around $\pb$, till $\pa$ becomes a boundary point --
    formally, let $\beta''$ be the minimum value of $\beta$ such that
    $C' = f_{\beta,\pb}(C_1)$ contains $\pa$. Since
    $C' \subseteq C_1 \subseteq C$, and $\pa, \pb \in \partial C'$,
    the claim follows.
\end{proof}

The following standard claim, usually stated for the standard Delaunay
triangulations, also holds for homothets.

\begin{claim}
    \clmlab{c:t:connected}%
   Let $\CC$ be a bounded close convex shape.  Given a set of points
   $\PS \subseteq \Re^2$ in general position with respect to $\CC$,
   let $\DG = \DG_{\CC}(\PS)$ be the $\CC$-Delaunay triangulation of
   $\PS$. For any homothet $C$ of $\CC$, we have that
   $\restrictY{\DG}{C}$ is connected.%
\end{claim}

\begin{proof}
    We prove that for any homothet $C$ with two points
    $\pa,\pb\in \PS$ on its boundary, there is a path between $\pa$
    and $\pb$ in $\restrictY{\DT}{C}$, and \lemref{shrink:shrank} will
    immediately imply the general statement. The proof is by induction
    over the number $m$ of points of $\PS$ in the interior of $C$. If
    $m=0$ then $C$ contains no points of $\PS$ in its interior, and
    thus $\pa \pb$ is an edge of the Delaunay triangulation, as $C$
    testifies.

    \begin{figure}[h]
        \phantom{}\hfill%
        \includegraphics[page=1]{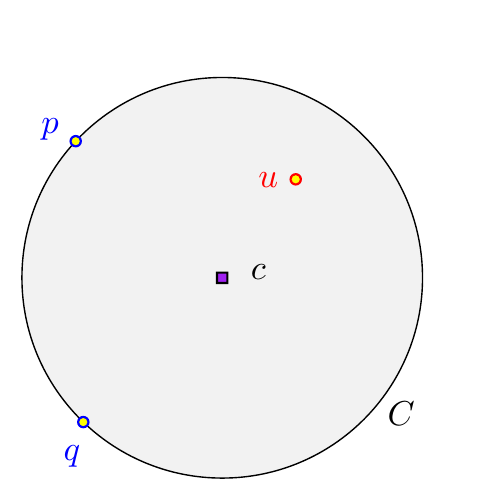}%
        \hfill%
        \includegraphics[page=2]{figs/shrink}%
        \hfill%
        \includegraphics[page=3]{figs/shrink}%
        \hfill%
        \phantom{}%
        \caption{An illustration of the proof of
           \clmref{c:t:connected} in the case that $C$ is a disk.}
        \figlab{shrink}
    \end{figure}

    Otherwise, let $\pc\in \PS$ be a point in the interior of
    $C$. From \lemref{shrink:shrank} we get that there exists a
    homothet $C'$ of $C$ with $C'\subseteq C$, such that $\pa$ and
    $\pc$ lie on the boundary of $C'$ Thus, by induction, there is a
    path $\gamma'$ between $\pa$ and $\pc$ in
    $\restrictY{\DT}{C'} \subseteq \restrictY{\DT}{C}$. Similarly,
    there must be a homothet $C''$, that gives rise to a path
    $\gamma''$ between $\pc$ and $\pb$, and concatenating the two
    paths results in a path between $\pa$ and $\pb$ in
    $\restrictY{\DT}{C}$.
\end{proof}

\subsection{The generic construction}
\seclab{disk_construction}

The input is a set $\PS$ of $n$ points in the plane (in general
position) with spread $\spread = \spreadX{\PS}$, and a parameter
$\eps \in (0,1)$. We have a convex body $\CC$ that defines the
``unit'' ball. The task is to construct a local spanner for any
homothet of $\CC$.

The algorithm computes a $(1/\epsA)$-\WSPD $\WS$ of $\PS$ using the
algorithm of \lemref{s:s:p:d:spread}, where $\epsA = \eps/6$.  For
each pair $\Pair = \{\PSX, \PSY \} \in \WS$, the algorithm computes
the $\CC$-Delaunay triangulation
$\DT_{\Pair} = \DG_{\CC}(\PSX \cup \PSY)$. The algorithm adds all the
edges in $\DT_{\Pair} \cap (\PSX \otimes \PSY)$ to the computed graph
$\G$.

\subsubsection{Analysis}

\myparagraph{Size} %
For each pair $\Pair = \{\PSX, \PSY\}$ in the \WSPD, its
$\CC$-Delaunay triangulation contains at most $\Of( |\PSX| + |\PSY|)$
edges. As such, the number of edges in the resulting graph is bounded
by
\begin{math}
    \sum_{\{\PSX, \PSY\} \in \WS} O\bigl( |\PSX| + |\PSY| \bigr) =%
    \Of\pth{ \WeightX{\WS} }%
    =%
    \Of\pth{ \frac{n\log \spread}{{\epsA}^{2}}},
\end{math}
by \lemref{s:s:p:d:spread}.

\myparagraph{Construction time} %
The construction time is bounded by
\begin{equation*}
    \sum_{\{\PSX, \PSY\} \in \WS} O\bigl( (|\PSX| + |\PSY|) \log
    (|\PSX| + |\PSY|) \bigr) =%
    \Of\pth{ \WeightX{\WS} \log n }%
    =%
    \Of\pth{ \frac{n\log \spread \log n}{{\epsA}^{2}}}.    
\end{equation*}

\begin{lemma}[Local spanner property]
    For $\PS, \CC, \eps$ as above, let $\G$ be the graph constructed
    above for the point set $\PS$. Then, for any homothet $C$ of $\CC$
    and any two points $\px, \py \in \PS \cap C$, we have that
    $\restrictY{\G}{C}$ has a $(1+\eps)$-path between $\px$ and
    $\py$. That is, $\G$ is a $\CC$-local $(1+\eps)$-spanner.
\end{lemma}

\begin{proof}
    Fix a homothet $C$ of $\CC$, and consider two points
    $\pa, \pb \in \PS \cap C$.  The proof is by induction on the
    distance between $\pa$ and $\pb$ (or more precisely, the rank of
    their distance among the $\binom{n}{2}$ pairwise distances).
    Consider the pair $\Pair = \{ \PSX, \PSY \}$ such that
    $\px \in \PSX$ and $\py \in \PSY$.

    \remove{%
       For the base case, consider the case that $\px$ is the
       nearest-neighbor to $\py$ in $\PS$, and $\py$ is the
       nearest-neighbor to $\px$ in $\PS$.  It must be, because of the
       separation property of $\Pair$, that $\PSX$ and $\PSY$ are
       singletons. Indeed, if $\PSX$ contains another point, then
       $\py$ would not be the nearest-neighbor to $\px$ (this is true
       for $\epsA < 1/2$). As such, $\px \py \in \DT_\Pair$,
       $\px, \py \in C$, and the edge $\px \py \in \EGX{\G}$, implying
       the claim.  }%
    
    If $\px \py \in \DT_\Pair$ then the claim holds, so assume this is
    not the case. By the connectivity of $\DT_\Pair \cap C$, see
    \clmref{c:t:connected}, there must be points
    $\px' \in \PSX \cap C$, $\py' \in \PSY \cap C$, such that
    $\px'\py' \in \EGX{ \DT_\Pair}$. As such, by construction, we have
    that $\px'\py' \in \EGX{\G}$. Furthermore, by the separation
    property, we have that
    \begin{equation*}
        \max \pth{ \diameterX{\PSX}, \diameterX{\PSY} }%
        \leq%
        \epsA \, \dsY{\PSX}{\PSY}%
        \leq%
        \epsA \ell,
    \end{equation*}
    where $\ell = \dY{\px}{\py}$. In particular,
    $\dY{\px'}{\px} \leq \epsA \ell$ and
    $\dY{\py'}{\py} \leq \epsA \ell$. As such, by induction, we have
    $\dGZ{\G}{\px}{\px'} \leq (1+\eps)\dY{\px}{\px'} \leq
    (1+\eps)\epsA \ell$ and
    $\dGZ{\G}{\py}{\py'} \leq (1+\eps)\dY{\py}{\py'} \leq
    (1+\eps)\epsA \ell$.  Furthermore,
    $\dY{\px'}{\py'} \leq (1+2\epsA)\ell$. As $\px'\py' \in \EGX{\G}$,
    we have
    \begin{align*}
      \dGZ{\G}{\px}{\py}%
      &\leq%
        \dGZ{\G}{\px}{\px'}%
        +\dY{\px'}{\py'}
        +
        \dGZ{\G}{\py'}{\py}%
        \leq%
        (1+\eps)\epsA \ell
        +(1+2\epsA)\ell
        + (1+\eps)\epsA \ell
        \leq%
        \pth{ 2\epsA +1+2\epsA + 2\epsA } \ell
      \\&%
      =%
      \pth{ 1+ 6\epsA  } \ell
      \leq%
      \pth{ 1+ \eps  } \dY{\px}{\py},
    \end{align*}
    if $\epsA \leq \eps/6$.
\end{proof}

\myparagraph{The result} %
We thus get the following.

\begin{theorem}
    \thmlab{main:1}%
    Let $\CC$ be a bounded convex shape in the plane, let $\PS$ be a
    given set of $n$ points in the plane (in general position), and
    let $\eps \in (0,1/2)$ be a parameter. The above algorithm
    constructs a $\CC$-local $(1+\eps)$-spanner $\G$. The spanner has
    $\Of\pth{ \eps^{-2} n\log \spread }$ edges, and the construction
    time is $\Of\pth{ \eps^{-2} n\log \spread \log n }$.  Formally,
    for any homothet $C$ of $\CC$, and any two points
    $\pa, \pb \in \PS \cap C$, we have a $(1+\eps)$-path in
    $\restrictY{\G}{C}$.
\end{theorem}

\subsubsection{Applications and comments}

The following defines a ``visibility'' graph when we are restricted to
a region $R$, where two points are visible if there is a witness
homothet contained in $R$ having both points on its boundary.
\begin{defn}
    Let $\CC$ be a bounded convex shape in the plane.  Given a region
    $R$ in the plane and a point set $\PS$, consider two points
    $\pa, \pb \in \PS$. The edge $\pa \pb$ is \emphw{safe} in $R$ if
    there is a homothet $C$ of $\CC$, such that
    $\pa,\pb \in C \subseteq R$. The \emphi{safe graph} for $\PS$ and
    $R$, denoted by $\GY{\PS}{R}$, is the graph formed by all the safe
    edges in $\PS$ for $R$. Note, that this graph might have a
    quadratic number of edges in the worst case.
\end{defn}

Observe that $\GY{\PS}{\Re^2}$ is a clique. Surprisingly, the spanner
graph described above, when restricted to region $R$, is a spanner for
$\GY{\PS}{R}$.

\begin{corollary}
    \corlab{safe:graph}
	
   Let $\CC$ be a bounded convex body, $\PS$ be a set of $n$ points in
   the plane, $\eps \in (0,1)$ be a parameter, and let $\G$ be a
   $\CC$-local $(1+\eps)$-spanner of $\PS$.
	
   Consider a region $R$ in the plane, and the associated graph
   $\GA = \GY{\PS}{R}$, we have that $\restrictY{\G}{R}$ is a
   $(1+\eps)$-spanner for $\GA$. Formally, for any two points
   $\pa, \pb \in \PS \cap R$, we have that
   $\dGZ{\restrictY{\G}{R}}{\pa}{\pb}\leq (1+\eps)\dGZ{\GA}{\pa}{\pb}
   $.
	
   In particular, for any convex region $D$, the graph $\G \gminus D$
   is a $(1+\eps)$-spanner for $\GY{\PS}{\Re^2} \gminus D$.
\end{corollary}
\begin{proof}
    Consider the shortest path $\pi = u_1 u_2 \ldots u_k$ between
    $\pa$ and $\pb$ realizing $\dGZ{\GA}{\pa}{\pb}$. Every edge
    $e_i = u_i u_{i+1}$ has a homothet $C_i$ such that
    $u_i, u_{i+1} \in C_i \subseteq R$. As such, there is a
    $(1+\eps)$-path between $u_i$ and $u_{i+1}$ in
    $\restrictY{\G}{C_i} \subseteq \restrictY{\G}{R}$. Concatenating
    these paths directly yields the desired result.
	
    The second claim follows by observing that the complement of $D$
    is the union of halfspaces, and halfspaces can be considered to be
    ``infinite'' homothets of $\CC$. As such, the above argument
    applies verbatim.
\end{proof}

\begin{remark}
    The above implies that local spanners for homothets are also
    robust to convex region faults. Namely, this construction both
    provides a local spanner and a fault tolerant spanner, where the
    locality is for homothets of the given shape, and the faults are
    for any convex regions.
\end{remark}

\subsection{Lower bounds}
\seclab{lower:bound}

\subsubsection{A lower bound for local spanner for disks}

The result of \thmref{main:1} is somewhat disappointing as it depends
on the spread of the point set (logarithmically, but still).  Next, we
show a lower bound proving that this dependency is unavoidable, even
in the case of disks.

\SoCGVer{\bigskip}%
\myparagraph{Some intuition} A natural way to attempt a
spread-independent construction is to try and emulate the construction
of Abam \etal \cite{abfg-rftgs-09} and use a \SSPD instead of a \WSPD,
as the total weight of the \SSPD is near linear (with no dependency on
the spread). Furthermore, after some post processing, one can assume
every pair $\Pair = \{ \PSX, \PSY \}$ is angularly $\eps$-separated --
that is, there is a double wedge with angle $\leq \eps$, such that
$\PSX$ and $\PSY$ are of different sides of the double wedge. The
problem is that for the local disk $\disk$, it might be that the
bridge edge between $\PSX$ and $\PSY$ that is in
$\DT_\Pair \cap \disk$ is much longer than the distance between the
two points of interest. This somewhat counter-intuitive situation is
illustrated in \figref{bad}.

\begin{figure}[h]
    \phantom{} \hfill%
    \includegraphics[page=1,width=0.3\linewidth]{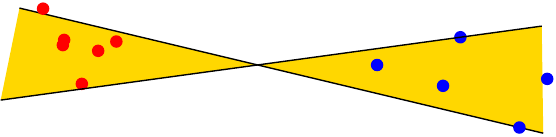}
    \hfill%
    \includegraphics[page=2,width=0.3\linewidth]{figs/bad_example}
    \hfill%
    \includegraphics[page=3,width=0.3\linewidth]{figs/bad_example}
    \hfill%
    \phantom{}
    \caption{A bridge too far -- the only surviving bridge between the
       red and blue points is too far to be useful if the sets of
       points are not well separated.}
    \figlab{bad}
\end{figure}

\SaveContent{\LemmaDiskLowerBound}%
{%
   For $\eps =1/4$, and parameters $n$ and $\spread \geq 1$, there is
   a point set $\PS$ of $n + \ceil{\log \spread}$ points in the plane,
   with spread $\Of( n \spread )$, such that any local
   $(1+\eps)$-spanner of $\PS$ for disks, must have
   $\Omega( n \log \spread )$ edges.  }

\begin{lemma}
    \lemlab{l:s:lower:bound}
    \LemmaDiskLowerBound

\end{lemma}

\begin{figure}[h]
    \centering%
    \includegraphics{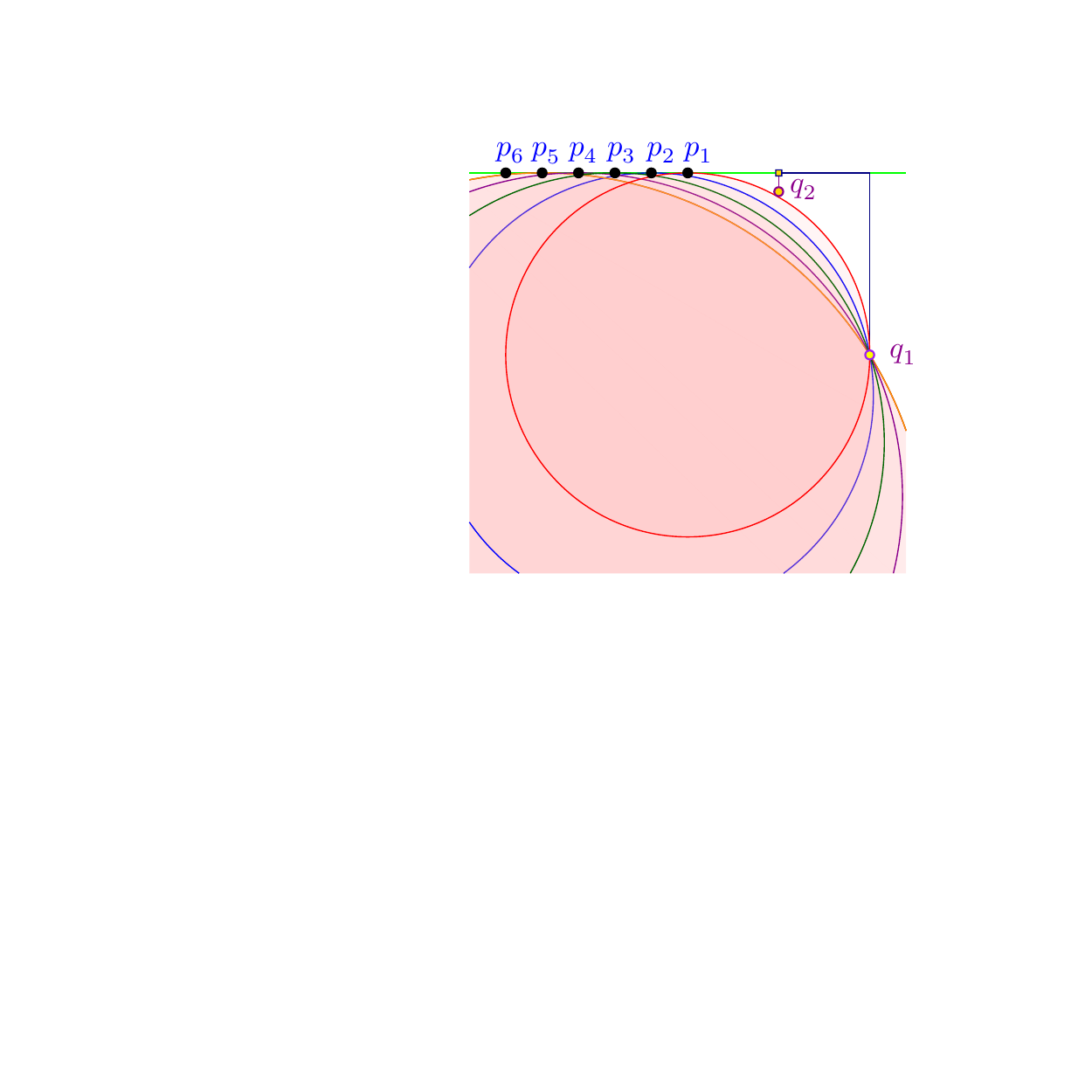}%
    \caption{The set of disks $D_1$, and the construction of $\pb_2$.}
    \figlab{flower}
\end{figure}

\begin{proof}
    Let $\pa_i = (-i, 0)$, for $i=1, \ldots, n$.  Let
    $M = 1+ \ceil{ \log_2 \spread }$ and $\pb_1 = (n 2^M, -1)$.  For a
    point $\pa$ on the $x$-axis, and a point $\pb$ below the $x$-axis
    and to the right of $\pa$, let $\diskVY{\pa}{\pb}$ be the disk
    whose boundary passes through $\pa$ and $\pb$, and its center has
    the same $x$-coordinate as $\pa$.
	
    In the $j$\th iteration, for $j=2,\ldots, M-1$, Let
    $x_j = n2^{M-j+1} = x(\pb_{j-1})/2$, and let $y_j < 0$ be the
    maximum $y$-coordinate of a point that lies on the intersection of
    the vertical line $x =x_j$ and the disks of
    $D_1 \cup \cdots \cup D_j$ where
	
    \begin{equation*}
        D_j = \Set{ \diskVY{\pa_i}{\pb_{j-1}} }{i=1,\ldots, n},
    \end{equation*}
    see \figref{flower} for an illustration of $D_1 $.
    
    Let $\pb_j = (x_j, 0.99y_j)$.

    Clearly, the point $\pb_j$ lies outside all the disks of
    $D_1 \cup \ldots \cup D_j$. The construction now continues to the
    next value of $j$. Let
    $\PS = \{ \pa_1, \ldots, \pa_n, \pb_2, \ldots, \pb_M \}$. We have
    that $|\PS| = n + M-1$.
	
    The minimum distance between any points in the construction is $1$
    (i.e., $\dY{\pa_1}{\pa_2}\bigr.$). Indeed $x(\pb_{M-1}) = 4n$ and
    thus $\dY{\pb_{M-1}}{\pa_1} \geq 2n$.  The diameter of $\PS$ is
    $\dY{\pa_1}{\pb_1} = \sqrt{(n +n2^{M})^2+ 1} \leq 2 n 2^M $. As
    such, the spread of $\PS$ is bounded by
    $\leq n 2^{M+1} = \Of( n \spread)$.
	
    For any $i$ and $j$, consider the disk
    $\diskVY{\pa_i}{\pb_j}$. This disk does not contain any point of
    $\pa_1,\ldots, \pa_{i-1}, \pa_{i+1}, \ldots, \pa_{n}$ since its
    interior lies below the $x$-axis. By construction it does not
    contain any point $\pb_{j+1}, \ldots, \pb_{M-1}$. This disk
    potentially contains the points $\pb_{j-1}, \ldots, \pb_1$, but
    observe that for any index $k \in \IRX{j-1}$, we have that
    \begin{equation*}
        \dY{\pa_i}{\pb_k} %
        =%
        \sqrt{ \pth{i + n 2^{M-k+1} }^2 + \bigl( y(\pb_j) \bigr)^2 },
    \end{equation*}
    which implies that
    \begin{math}
        n 2^{M-k+1}%
        \leq%
        \dY{\pa_i}{\pb_k} %
        <%
        n( 2^{M-k+1} +2).
    \end{math}
    We thus have that
    \begin{equation*}
        \frac{\dY{\pa_i}{\pb_k}}{\dY{\pa_i}{\pb_j}}
        \geq%
        \frac{n 2^{M-k+1}}{n( 2^{M-j+1} +2)}
        =%
        \frac{ 2^{M-j}\cdot 2^{j-k}}{ 2^{M-j} +1}
        =%
        \frac{  2^{j-k}}{ 1 + 1/2^{M-j}}
        \geq
        \frac{2}{1+1/2}
        = \frac{4}{3}
        > 1+\eps,
    \end{equation*}
    since $j \in \IRX{M-1}$.  Namely, the shortest path in $\G$
    between $\pa_i$ and $\pb_j$, can not use any of the points
    $\pb_{1}, \ldots \pb_{j-1}$. As such, the graph $\G$ must contain
    the edge $\pa_i \pb_j$. This implies that
    $|\EGX{\G}| \geq n (M-1)$, which implies the claim.
\end{proof}

\subsubsection{A lower bound for triangles}
 
\begin{lemma}%
    \lemlab{l:b:triangles}%
    For any $n > 0$, and $\Phi = \Omega(n)$, one can compute a set
    $\PS$ of $n+ \Of(\log\Phi)$ points, with spread $\Of(\Phi n)$, and
    a triangle $\triangle$, such that any $\triangle$-local
    $(3/2)$-spanner of $\PS$ requires $\Omega\pth{n\log\Phi }$ edges.
\end{lemma}

\begin{figure}[b]
    \centering \includegraphics{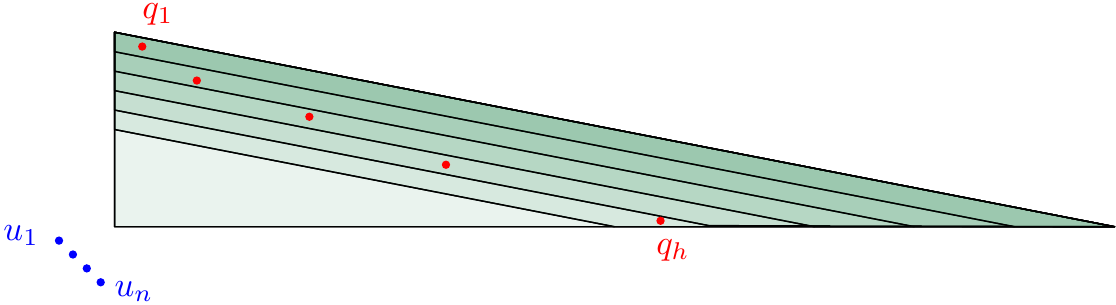}
    \caption{An Illustration of the construction of
       \lemref{l:b:triangles}.}
    \figlab{tri_low_bd}
\end{figure}

\begin{proof}
    Let $h = \ceil{\log_2 \Phi}$.  Let $\triangle$ be the triangle
    formed by the points $(0,0)$, $(0,1)$ and $(8\Phi h,0)$.  The
    hypotenuse of this triangle lies on the line
    $\Line \equiv \frac{1}{8\Phi h}x + y = 1$, and let
    $v = \bigl(\frac{1}{8\Phi h}, 1\bigr)$ be the vector orthogonal to
    this line.
	
    For $i \in \IRX{h}$ and $j \in \IRX{n}$, let
    \begin{align*}
      \pb_i = \bigl(2^{i+1}, 1 - i/h \bigr)
      \qquad \text{ and }\qquad
      \pc_j = \bigl(\tfrac{j}{n}-1, -\tfrac{j}{n} \bigr),
    \end{align*}
    and let $\PS = \{ \pb_1, \ldots, \pb_h, \pc_1, \ldots, \pc_n\}$,
    see \figref{tri_low_bd}.  Observe that
    $\cpX{\PS} = \dY{\pc_1}{\pc_2} = \sqrt{2}/n$, and as such we have
    that
    $\spreadX{\PS} = n \cdot \diamX{\PS}/\sqrt{2} \leq n(4\Phi + 2n)
    \leq 8 \Phi n$, as $\Phi \geq n$.  Observe that
    \begin{equation*}
	\DotProd{\pb_{i+1} - \pb_i}{v}%
	=%
	\DotProd{(2^{i+1},-\tfrac{1}{h})}{ \bigl(\tfrac{1}{8\Phi h},
           1\bigr)}
	\leq
	\tfrac{4\Phi}{8\Phi h} - \tfrac{1}{h} < 0.
    \end{equation*}
    That is, the points $\pb_1, \ldots, \pb_i$ are increasing in
    distance from $\Line$.
	
    Let $\triangle_{i,j}$ be the homothet of $\triangle$, that has its
    bottom left corner at $\pc_j$, and its hypotenuse passes through
    $\pb_i$. By the above,
    $\PS(i,j) = \triangle_{i,j} \cap \PS = \{ \pc_j, \pb_i, \pb_{i+1},
    \ldots \pb_h \}$.  Any $(1+\eps)$-spanner for $\PS(i,j)$ must
    contain the edge $\pc_j \pb_i$. Indeed, we have, for any $k$, that
    $2^{k+1} \leq \dY{\pc_j}{\pb_{k}} \leq 2^{k+1} +3$. As such, any
    path on a graph induced on $\PS(i,j)$ from $\pc_j$ to $\pb_i$ that
    uses (say) a midpoint $\pb_k$, for $k >i$, must have dilation at
    least
    \begin{equation*}
	\frac{\dY{\pc_j}{\pb_k} +
           \dY{\pb_k}{\pb_i}}{\dY{\pc_j}{\pb_{i}}}
	\geq%
	\frac{2^{k+1} + 2^{k} }{2^{i+1} + 3}
	\geq%
	\frac{3 \cdot 2^{i+1} }{(1+3/4) 2^{i+1} }
	=
	\frac{12}{7}
	>
	\frac{3}{2}.
    \end{equation*}
	
    Thus, any $\triangle$-local $3/2$-spanner for homothets of
    $\triangle$, must contain the edge $\pb_i \pc_j$, for any
    $i \in \IRX{h}$ and $j \in \IRX{n}$. Thus, such a spanner must
    have $ \Omega( n \log \Phi)$ edges, as claimed.
\end{proof}

\subsection{Local spanners for fat triangles}
\seclab{triangles}

While local spanners for homothets of an arbitrary convex shape are
costly, if we are given a triangle $\triangle$ with the single
constraint that $\triangle$ is not too ``thin'', then one can
construct a $\triangle$-local $t$-spanner with a number of edges that
does not depend on the spread of the points. See \figref{tri_low_bd}
for an illustration of a construction showing that dependency if
"thin" triangles are allowed.

\begin{defn}
    A triangle $\triangle$ is \emphi{$\alpha$-fat} if the smallest
    angle in $\triangle$ is at least $\alpha$.
\end{defn}

\begin{figure}[t]
    \centering \phantom{}%
    \hfill%
    \includegraphics[page=1]{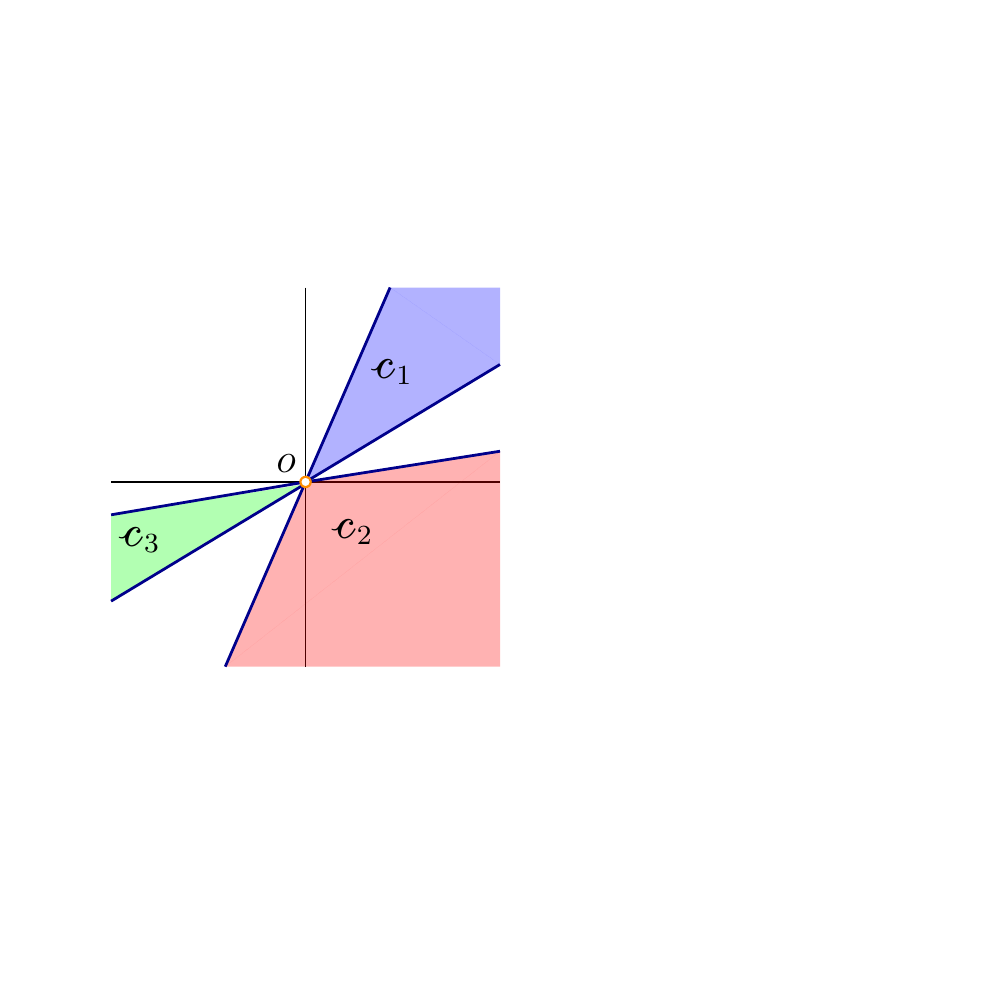}%
    \hfill%
    \includegraphics[page=2]{figs/triangle_cones}%
    \hfill%
    \phantom{}%
    \caption{For the triangle $\triangle$ with angles
       $\alpha_1,\alpha_2$, and $\alpha_3$ we create the cones
       $\cone_1,\cone_2$, and $\cone_3$.}
    \figlab{tri_cones}
\end{figure}

\subsubsection{Construction}
\seclab{tri_construction}

The input is a set $\PS$ of $n$ points in the plane, an $\alpha$-fat
triangle $\triangle$, and an approximation parameter $\eps \in
(0,1)$. Let $v_i$ denote the $i$\th vertex of $\triangle$, $\alpha_i$
be the adjacent angle, and let $e_i$ denote the opposing edge, for
$i\in \IRX{3}$.  Let
$\cone_i = \Set{ (\pa - v_i)t }{ \pa \in e_i \text{ and } t \geq 0}$
denote the cone with an apex at the origin induced by the $i$\th
vertex of $\triangle$.  Let $\dir_i$ be the outer normal of
$\triangle$ orthogonal to $e_i$.  See \figref{tri_cones} for an
illustration. Let $\ConeSet_i$ be a minimum size partition of
$\cone_i$ into cones each with angle in the range $[\beta/2, \beta]$,
where $\beta = \eps \alpha/\gamma$, and $\gamma>1$ is a constant to be
determined shortly.  For each point $\pa\in\PS$, and a cone
$\cone \in \ConeSet_i$, let $\nnZ{i}{\pa}{\cone}$ be the first point
in $(\PS - \pa)\cap ( \pa + \cone)$ ordered by the direction $\dir_i$
(it is the ``nearest-neighbor'' to $\pa$ in $\pa+\cone$ with respect
to the direction $\dir_i$).

\paragraph*{The construction}
Let $\G$ be the graph over $\PS$ formed by connecting every point
$\pa \in \PS$ to $\nnZ{i}{\pa}{\cone}$, for all $i \in \IRX{3}$ and
$\cone \in \ConeSet_i$.

\subsubsection{Analysis}

\begin{lemma}
    \lemlab{cone_edge:triangles}%
   Let $\pa \in \PS$, $\cone \in \ConeSet_i$, and
   $\pc =\nnZ{i}{\pa}{\cone}$, and let $\pb$ be a point in
   $(\PS \cap (\pa + \cone)) \setminus \{\pa, \pc\}$.  We have that
   \begin{math}
       \dY{\pa}{\pc}%
       +%
       (1+\eps)\dY{\pb}{\pc}%
       \leq%
       (1+\eps)\dY{\pa}{\pb}
   \end{math}
   and
   \begin{math}
       \dY{\pb}{\pc}%
       \leq%
       \dY{\pa}{\pb}.
   \end{math}
\end{lemma}
\begin{figure}[ht]
    \centering%
    \includegraphics{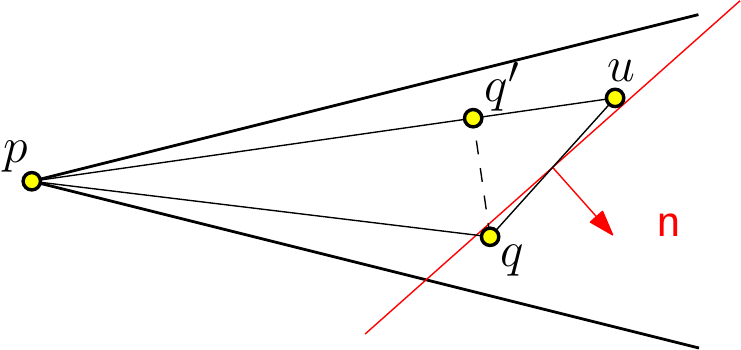}
    \caption{The case that $\dY{\pa}{\pb}\leq \dY{\pa}{\pc}$ in
       \lemref{cone_edge:triangles}. The vector used to determine
       $\nnZ{i}{\pa}{\cone}$ is shown in red, and denoted $\dir$}
    
    \figlab{tri_cone_edge:a}
\end{figure}
    
\begin{proof}
    Consider the triangle $\Delta\pa\pb\pc$ and denote the angles at
    $\pa,\pb,$ and $\pc$ by $\angleX{\pa},\angleX{\pb}$, and
    $\angleX{\pc}$ respectively.  Since the angle of $\cone$ is
    smaller than $60$ degrees (for an appropriate choice of $\gamma$),
    we have that
    $\dY{\pb}{\pc}\leq \max \{\dY{\pa}{\pc},\dY{\pa}{\pb}\}$.

    Consider the case that $\dY{\pa}{\pb}\leq \dY{\pa}{\pc}$,
    illustrated in \figref{tri_cone_edge:a}.  Observe that
    $\angleX{\pc} \leq \angleX{\pb}$.  As such
    $\angleX{\pc} \leq \pi/2$.  Furthermore,
    $\angleX{\pc} \geq \alpha \gg \eps \alpha/\gamma = \beta \geq
    \angleX{\pa}$. Similarly,
    $\angleX{\pb} \in [\alpha, \pi -\alpha]$.  By the $1$-Lipshitz of
    $\sin$, and as $\sin x \approx x$, for small $x$, and for $\gamma$
    sufficiently large, we have that
    \begin{equation*}
        \sin( \angleX{\pb} + \angleX{\pa}) \in [1-\eps/4, 1+\eps/4] \sin
        \angleX{\pb}%
        \qquad\text{and}\qquad
        \sin \angleX{\pa} \leq (\eps/4) \sin \angleX{\pc}.
    \end{equation*}
    As such, by the law of sines, we have that
    \begin{math}
        \frac{\dY{\pb}{ \pc}}{\sin{\angleX{\pa}}}%
        =%
        \frac{\dY{\pa}{ \pb}}{\sin{\angleX{\pc}}}%
        =%
        \frac{\dY{\pa}{ \pc}}{\sin{\angleX{\pb}}}.%
    \end{math}
    This implies that
    \begin{equation*}
        \dY{\pa}{\pc}%
        +%
        (1+\eps)\dY{\pb}{\pc}%
        =%
        \pth{\frac{\sin{\angleX{\pb}}}{\sin{\angleX{\pc}}} +(1+\eps)
           \frac{\sin{\angleX{\pa}}}{\sin{\angleX{\pc}}} }\dY{\pa}{ \pb}.        
    \end{equation*}
    Observe, by the above that
    \begin{equation*}
        \frac{\sin{\angleX{\pb}}}{\sin{\angleX{\pc}}}
        +(1+\eps)
        \frac{\sin{\angleX{\pa}}}{\sin{\angleX{\pc}}}        
        \leq
        \frac{\sin{\angleX{\pb}}}{\sin{\pth{\angleX{\pa} + \angleX{\pb}}}}
        +(1+\eps)
        \frac{\eps}{4}
        \leq%
        \frac{\sin{\angleX{\pb}}}{(1-\eps/4)\sin{\pth{\angleX{\pb}}}}
        +(1+\eps)
        \frac{\eps}{4}
        \leq%
        1+\eps.
    \end{equation*}
    
    \begin{figure}[h]
        \centering%
        \includegraphics[page=2]{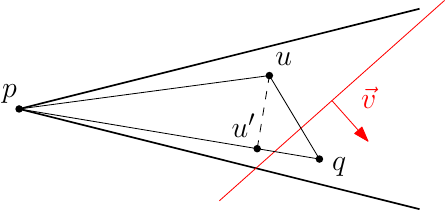}
        \caption{The case that $\dY{\pa}{\pb}> \dY{\pa}{\pc}$ in
           \lemref{cone_edge:triangles}.}%
        \figlab{tri_cone_edge:b}
    \end{figure}
    
    The other possibility is that $\dY{\pa}{\pb} > \dY{\pa}{\pc}$,
    illustrated in \figref{tri_cone_edge:b}.  Let $\pc'$ be the
    projection of $\pc$ to $\pa \pb$. Observe that
    \begin{equation*}
        \dY{\pc}{\pc'}%
        =%
        \dY{\pa}{\pc'} \tan \angleX{\pa}%
        \leq%
        2\beta\dY{\pa}{\pc'} \leq (\eps/8) \dY{\pa}{\pc'}.
    \end{equation*}
    Observe that
    $\cos \angleX{\pa} \geq 1 - (\angleX{\pa})^2/2 \geq 1-\eps^2/8$ as
    $\angleX{\pa}$ is an angle smaller than (say) $\eps/16$.  As such
    $1/\cos \angleX{\pa} \leq 1+\eps^2/4$.  This implies that
    \begin{math}
        \dY{\pa}{\pc}%
        \leq%
        {\dY{\pa}{\pc'}} / {\cos \angleX{\pa} } \leq %
        (1+\eps^2/4)\dY{\pa}{\pc'}.
    \end{math}
    We thus have that
    \begin{align*}
      \tau
      &=%
        \dY{\pa}{\pc}%
        +%
        (1+\eps)\dY{\pb}{\pc}%
        \leq%
        (1+\eps^2/4)\dY{\pa}{\pc'}
        +(1+\eps)\pth{ \Bigl.
        \dY{\pc}{\pc'} + \dY{\pc'}{\pb}}
      \\&
      \leq%
      \bigl(1+\eps^2/4 + (1+\eps)\eps/8\bigr)\dY{\pa}{\pc'} 
      + (1+\eps)\dY{\pc'}{\pb}
      \leq%
      (1+\eps)\dY{\pa}{\pb}.
    \end{align*}
\end{proof}

\begin{lemma}
    \lemlab{shrink:triangles}%
    Let $\triangle$ be a triangle that contains two points $\pa,
    \pb$. Then, there is a homothet $\triangle'\subseteq \triangle$ of
    $\triangle$, such that one of these points is a vertex of
    $\triangle'$, and the other point lies on a facing edge of
    $\triangle'$. %
\end{lemma}

\begin{proof}
    This follows by the same shrinking argument as
    \lemref{shrink:shrank}, with the addition of a single step. When a
    homothet $\triangle'$ with $\pa,\pb \in \partial\triangle'$ is
    found, if neither point is on a vertex, we ``push'' the only edge
    that does not contain one of the points towards the vertex $v$
    opposite of it (this the same mapping described in
    \lemref{shrink:shrank} with center $v$), until one of the points,
    say $\pa$ lies on the edge. $\pa$ now lies on two edges, meaning,
    at a vertex, while $q$ lies on the only remaining edge which must
    be opposite of that vertex. See \figref{shrink:triangle}.
\end{proof}

\begin{figure}[ht]
    \includegraphics[page=1,width=0.22\linewidth]{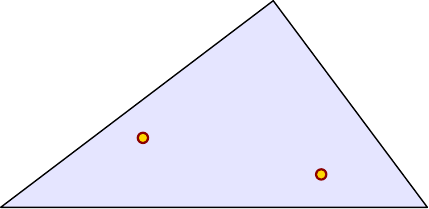}%
    \hfill
    \includegraphics[page=2,width=0.22\linewidth]{figs/narrow_triangle}%
    \hfill
    \includegraphics[page=3,width=0.22\linewidth]{figs/narrow_triangle}%
    \hfill
    \includegraphics[page=4,width=0.22\linewidth]{figs/narrow_triangle}
    \caption{ An illustration of the shrinking process of
       \lemref{shrink:triangles}. The three left figures illustrates
       the process of \lemref{shrink:shrank}, for the case that the
       convex region $\CC$ is a triangle, and the rightmost figure is
       the additional final step.}
    \figlab{shrink:triangle}
\end{figure}

\paragraph*{Local spanner property}
\begin{lemma}
    \lemlab{local_span:triangles}%
    Let $\triangle'$ be a homothet of $\triangle$. For any two points
    $\pa, \pb \in \PS \cap \triangle'$, we have a $(1+\eps)$-path in
    $\G' = \restrictY{\G}{\triangle'}$.
\end{lemma}

\begin{proof}
    Consider the closest pair $\pa,\pb \in \PS \cap \triangle$. They
    must be connected directly in $\G'$, as otherwise there is a point
    $\pc \in \PS' = \PS \cap \triangle'$ in the cone containing the
    segment $\pa\pb$, such that $\pa \pc \in \EGX{\G'}$. But then, by
    \lemref{cone_edge:triangles}, we have
    \begin{math}
        \dY{\pa}{\pc}%
        +%
        (1+\eps)\dY{\pb}{\pc}%
        \leq%
        (1+\eps)\dY{\pa}{\pb},
    \end{math}
    which implies that either $\pa\pc$ or $\pb\pc$ are the closest
    pair, which is a contradiction.

    For any other pair $\pa,\pb\in \PS'$ we have from
    \lemref{shrink:triangles} that there exists a homothet
    $\triangle''\subseteq \triangle'$ with one of the two points, say
    $\pa$, at a vertex, and the other on the opposite edge. We
    therefore have a cone $\cone$ with apex at $\pa$ such that
    $\pb\in \cone\cap \triangle''$. If $\pa\pb$ is an edge in $\G$
    then we are done. Otherwise, we have a vertex $\pc\in \cone$ such
    that $\pa\pc$ is an edge in $\G$, and by
    \lemref{cone_edge:triangles} we have
    $\dY{\pb}{\pc}\leq \dY{\pa}{\pb}$, which, by induction, means that
    there exists a $(1+\eps)$ path between $\pc$ and $\pb$ in
    $\G$. \lemref{cone_edge:triangles} now implies that
    $\dY{\pa}{\pc}+(1+\eps)\dY{\pb}{\pc}\leq(1+\eps)\dY{\pa}{\pb}$. Thus,
    there is a $(1+\eps)$ path between $\pa$ and $\pb$ in $\G'$, as
    stated.
\end{proof}

\paragraph*{Size and running time}

\begin{theorem}
    \thmlab{l:s:triangle}%
    Let $\PS$ be a set of $n$ points in the plane, and let
    $\eps \in (0,1)$ be an approximation parameter. The above
    algorithm computes a $\triangle$-local $(1+\eps)$-spanner $\G$ for
    an $\alpha$-fat triangle $\triangle$.  The construction time is
    $\Of\pth{(\alpha\eps)^{-1}n\log n}$, and the spanner $\G$ has
    $\Of\pth{(\alpha\eps)^{-1} n\bigr.}$ edges.
\end{theorem}

\begin{proof}
    The local-spanning property is proven in
    \lemref{local_span:triangles}, and we are only left with bounding
    the size and the running time of the algorithm. The bound on the
    size is immediate from the construction, as every point $\pa$ is
    the apex of $\Of\pth{ \frac{2\pi}{\eps\alpha} }$ cones, each
    giving rise to a single edge incident to $\pa$.  The construction
    time is bounded by the construction time for a $\theta$-graph with
    cone size $\alpha\eps$, which is
    $\Of\pth{ (\alpha\eps)^{-1}n\log n}$ (\cite{c-aaspmp-87}).
\end{proof}

\subsection{A local spanner for nice polygons}
\seclab{k:gons}

\subsubsection{A good jump for narrow trapezoids}
\seclab{traps:jump}

As a reminder, a trapezoid is a quadrilateral with two parallel edges,
known as its \emph{bases}. The other two edges are its \emph{legs}.
For $\eps \in (0,1/4)$, a trapezoid $\Trap$ is \emphi{$\eps$-narrow}
if the length of each of its legs is at most
$\eps \cdot \diamX{\Trap}$.

\begin{lemma}
    \lemlab{good:jump:traps}%
   Let $\eps \in (0,1)$ be some parameter, and $\epsA = \eps/16$.  Let
   $\PX,\PY$ be two points sets that are $\epsA$-semi separated and
   $\epsA$-angularly separated (see \defref{angular:separated}), and
   let $\Trap$ be a $\epsA$-narrow trapezoid, with two points
   $\pa \in \PX$ and $\pb \in \PY$ lying on the two legs of
   $\Trap$. Then, one can compute a homothet $\Trap' \subseteq \Trap$
   of $\Trap$, such that:
   \begin{compactenumI}
       \item There are two points $\pa' \in \PX$ and $\pb' \in \PY$,
       such that $\pa'\pb'$ is an edge of the $\Trap$-Delaunay
       triangulation of $\PX \cup \PY$.
		
       \item We have that
       $(1+\eps)\dY{\pa}{\pa'} + \dY{\pa'}{\pb'} + (1+\eps)
       \dY{\pb'}{\pb} \leq (1+\eps) \dY{\pa}{\pb}$.
   \end{compactenumI}
\end{lemma}

\begin{figure}[ht]
    \phantom{}\hfill%
    \includegraphics[width=0.3\linewidth]{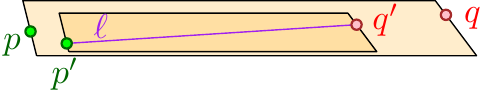}%
    \hfill%
    \includegraphics[page=2,width=0.3\linewidth]{figs/narrow_trap}%
    \hfill%
    \includegraphics[page=3,width=0.3\linewidth]{figs/narrow_trap}%
    \hfill\phantom{}%
    
    \phantom{}\hfill%
    (i)\qquad\qquad\qquad \hfill%
    (ii) \hfill%
    \qquad \qquad \qquad(iii) \hfill\phantom{}%
    
    \caption{Illustration of the settings in the proof of
       \lemref{good:jump:traps}. Left: A $\epsA$-narrow trapezoid with
       $\pa$ and $\pb$ on its legs. Center: $\pa$ and $\pb$ are
       $\epsA$-semi separated and $\epsA$-angularly separated. Right:
       The triangle of all the points of the trapezoids that their
       nearest point on $\pa \pb$ is $\pb$. }
    \figlab{narrow:trap}
\end{figure}
\begin{proof}
    Let $\DT = \DG_\Trap(\PX \cup \PY)$.  \clmref{c:t:connected}
    implies that $\DT \cap \Trap$ is connected. Thus, there is a path
    in $\DT \cap \Trap$ between $\pa$ and $\pb$, and thus, there must
    be an edge $\pa'\pb'$ along this path with $\pa' \in X$ and
    $\pb' \in Y$. This implies part (I).
    
    Let $\ell = \dY{\pa'}{\pb'}$. Assume for concreteness that
    $\dY{\pa}{\pa'} \leq \diamX{\PX} \leq \epsA \dsY{\PX}{\PY} \leq
    \epsA \ell\leq \epsA \diamC$, where $\diamC=\diamX{\Trap}$. Let
    $\pb''$ be the closest point on $\pa\pb$ to $\pb'$.
    
    We first consider the case that $\pb'' \in \interiorX{\pa\pb}$.
    We have that
    \begin{equation*}
        \dY{\pa}{\pb''}
        =%
        \dY{\pa}{\pb'} \cos \angle \pb' \pa \pb
        \geq
        \bigl(\dY{\pa'}{\pb'} - \dY{\pa}{\pa'} \bigr)
        \cos \angle \pb' \pa \pb
        \geq
        (1-\epsA) \ell \cdot (1-\epsA^2/2)
        \geq
        (1-2\epsA) \ell,
    \end{equation*}
    since $\cos \epsA \geq 1-\epsA^2 /2$, for $\epsA <1/2$.  Similar
    argumentation implies that $\dY{\pa}{\pb''} \leq (1+\epsA)
    \ell$. As such, we have
    \begin{equation*}
        \dY{\pb'}{\pb''} \leq ( 1+\epsA)\ell \sin \angle\pa'\pa \pb'
        \leq
        2 \epsA \ell.
    \end{equation*}
    Thus, we have that
    \begin{equation*}
        \dY{\pb}{\pb'}
        \leq
        \dY{\pb}{\pb''}  + \dY{\pb''}{\pb'}
        \leq%
        \dY{\pa}{\pb} -    \dY{\pa}{\pb''} + 2\epsA \ell
        \leq%
        \dY{\pa}{\pb} - (1-2\epsA)\ell + 2\epsA \ell
        \leq 
        \dY{\pa}{\pb} - \ell.
    \end{equation*}
    
    Thus, we have that
    \begin{align*}
      &(1+\eps)\dY{\pa}{\pa'} + \dY{\pa'}{\pb'} + (1+\eps)
        \dY{\pb'}{\pb}
        \leq%
        (1+\eps)\epsA \ell
        + \ell + (1+\eps)\bigl(
        \dY{\pa}{\pb} - \ell\bigr)
      \\&%
      \qquad=%
      (1+\eps)\dY{\pa}{\pb}
      +
      (1+\eps)\epsA \ell
      + \ell - (1+\eps)\ell
      \leq%
      (1+\eps)\dY{\pa}{\pb},
    \end{align*}
    for $\epsA \leq \eps/2$. Which establish the claim in this case.
    
    The case that $\pb''=\pa$ is impossible, because of the angular
    separation property. Thus, the only remaining possibility is that
    $\pb'' = \pb$. This however implies that $\pb'$ must be in the
    triangle of all the points of the trapezoids that their nearest
    point on $\pa \pb$ is $\pb$. The diameter of this triangle is
    bounded by the length of the leg of the trapezoid, which is
    bounded by $\epsA \, \diamC$. Namely, we have
    $\dY{\pb}{\pb'} \leq \epsA \diamC$. Similarly, we have
    \begin{math}
        (1 - 2\epsA) \diamC%
        \leq%
        \dY{\pa}{\pb} \leq%
        (1+2\epsA)\diamC.
    \end{math}
    Since $\dY{\pa}{\pa'}, \dY{\pb}{\pb'} \leq \epsA \diamC$, it
    follows that
    \begin{equation*}
        (1-4\epsA) \diamC%
        \leq%
        \ell%
        \leq%
        (1+4\epsA)\diamC.
    \end{equation*}
    As such, for $\epsA \leq \eps/8$ and $\eps \leq 1$, we have
    \begin{equation*}
        (1+\eps)\dY{\pa}{\pa'} + \ell + (1+\eps)
        \dY{\pb'}{\pb}
        \leq%
        4 \epsA \diamC + (1+4\epsA)\diamC
        =%
        (1+8\epsA)\diamC
        \leq
        (1+\eps) \dY{\pa}{\pb}.        
    \end{equation*}
\end{proof}

\subsubsection{Breaking a nice polygon into narrow %
   trapezoids}

For a convex polygon $\Body$, its \emphi{sensitivity}, denoted by
$\senseX{\Body}$, is the minimum distance between any two non-adjacent
edges (this quantity is no bigger than the length of the shortest edge
in the polygon).  A convex polygon $\Body$ is \emphi{$t$-nice}, if the
outer angle at any vertex of the polygon is at least $2\pi/t$, and the
length of the longest edge of $\Body$ is $\Of(\senseX{\Body})$.  As an
example, a $k$-regular polygon is $k$-nice.

\begin{lemma}
    \lemlab{narrow:traps:decomp}%
    Let $t$ be a positive integer.  Given a $t$-nice polygon $\Body$,
    and a parameter $\epsA$, one can cover it by a set $\Traps$ of
    $\Of(t^4/\epsA^3)$ $\epsA$-narrow trapezoids, such that for any
    two points $\pa, \pb \in \partial \Body$ that belong to two edges
    of $\Body$ that are not adjacent, there exists a narrow trapezoid
    $\Trap \in \Traps$, such that $\pa$ and $\pb$ are located on two
    different short legs of $\Trap$.%
\end{lemma}

\begin{proof}
    We show a somewhat suboptimal but simple construction. A $t$-nice
    polygon has at most $t$ edges.  Let $\psi$ be the sensitivity of
    $\Body$, and place a minimum set of points $\PS$ on the boundary
    of $\Body$, which includes all the vertices of $\Body$, and such
    that the distance between any consecutive pair of points is in the
    range $[\constA, 2\constA]$, where
    $\constA = \epsA \psi/ \constB$, for some sufficiently large
    constant $\constB$. In particular, let
    $M =\max_{e \in \EGX{\Body}} \ceil{\| e\| / \constA} = \Of(
    1/\epsA)$.
    
    In addition, place $\constC t$ equally spaced points between any
    two consecutive points of $\PS$, where $\constC$ is a constant to
    be determined shortly. Let $\PSA$ be the set resulting from $\PS$
    after adding all these points.
    
    \begin{figure}[h]
        \phantom{}%
        \hfill%
        \includegraphics[page=2,width=0.3\linewidth]{figs/decompose}
        \hfill%
        \includegraphics[page=3,width=0.3\linewidth]{figs/decompose}
        \hfill%
        \includegraphics[page=1,width=0.3\linewidth]{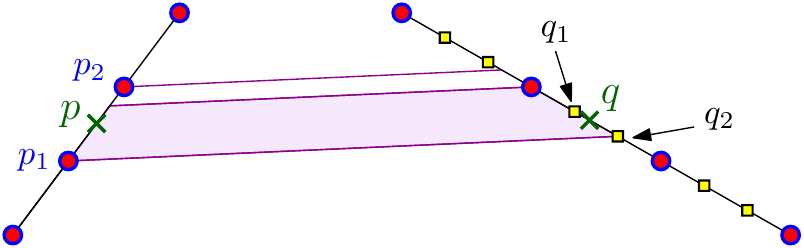}
        \hfill%
        \phantom{}%
        \caption{The points of $\PS$ (round), and all the points added
           to $\PS$ in order to create $\PSA$ (square). On the right,
           a ``vertical'' decomposition induced by one of the
           directions of $\PS \times \PSA$.}
        \figlab{v:decompose}
    \end{figure}
    
    We have that $|\PS| = \Of( t /\epsA)$ and
    $|\PSA| = \Of(t^2/\epsA)$. For a direction $v$, let $\Traps_v$ be
    the decomposition into trapezoids formed by shooting rays from
    inside $\Body$ in the direction of $v$ (or $-v$) from all the
    points of $\PS$, see \figref{v:decompose}. Let $\Traps_v'$ be the
    set resulting from throwing away trapezoids with legs that lie on
    adjacent edges.  It is easy to verify that all the trapezoids of
    $\Traps_v'$ are $\epsA$-narrow.  Let $U$ be the set of all
    directions induced by pairs of points of $\PS \times \PSA$, and
    let $\Traps = \cup_{u \in U} \Traps_u'$. We have that
    $|\Traps| = \Of( |\PS| \cdot |U| ) = \Of(|\PS|^2 |\PSA| ) =\Of(
    t^4 /\epsA^3)$.
    
    Consider any two points $\pa, \pb$ on non-adjacent edges of
    $\Body$, and let $\pa_1, \pa_2 $ be the two adjacent points of
    $\PS$ such that $\pa \in \pa_1 \pa_2$.  Now, let $\pb_1, \pb_2$ be
    the adjacent points of $\PSA$ such that $\pb \in \pb_1 \pb_2$.  We
    assume that $\pa_1, \pa_2, \pb_1,\pb_2$ are in this clockwise
    order along the boundary of $\Body$.
    
    Observe that when we project the interval $\pa_1 \pa_2$, to the
    line induced by $\pb_1 \pb_2$, in the direction
    $\overrightarrow{\pa_1 \pb_2}$, the projected interval contains
    $\pb_1 \pb_2$.  The last claim is intuitively obvious, but
    requires some work to see formally. The minimum height of a
    triangle involving three vertices of $\Body$ is formed by three
    consecutive vertices. In the worst case, this is an isosceles
    triangle with sidelength $\psi$ and base angle $\pi/t$. As such,
    the height of such a triangle is
    $h = \psi \sin( \pi/t) \geq \psi/t$.

    \begin{figure}[h]
        \includegraphics[page=2,width=0.23\linewidth]{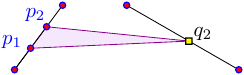}%
        \hfill%
        \includegraphics[page=3,width=0.23\linewidth]{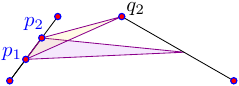}
        \hfill%
        \includegraphics[page=4,width=0.23\linewidth]{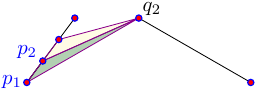}
        \hfill%
        \includegraphics[page=5,width=0.23\linewidth]{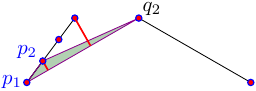}
        \caption{The height of the triangle
           $\triangle \pa_1 \pa_2 \pb_2$ is minimized as $\pb_2$ and
           $\pa_1$ are moved to vertices of $\Body$.}%
        \figlab{points:move}
    \end{figure}
    
    The height of the triangle $\triangle \pa_1 \pa_2 \pb_2$ is
    minimized when $\pa_1$ or $\pa_2$ is a vertex of $\Body$, and
    $\pb_2$ is at a vertex of $\Body$, see
    \figref{points:move}. Assume, for concreteness, that $\pa_1$ is a
    vertex of $\Body$, and observe that
    $\dY{\pa_1}{\pa_2} \geq \| e\|/M$, where $e$ is the edge of
    $\Body$ containing this segment. Using similar triangles, it is
    straightforward to show that the height of this triangle is at
    least $h' = h/M = \Omega( \eps \psi/t )$. The quantity $h'$ is a
    lower bound on the length of the projection of $\pa_1 \pa_2$ on
    the line spanned by $\pb_1 \pb_2$. However,
    $\dY{\pb_1}{\pb_2} \leq 2\constA/\constC t = \Of( \epsA \psi
    /\constC t) < h'$, by picking $\constC$ to be sufficiently large
    constant.
    
    This readily implies that the trapezoid induced by the direction
    $ u = \overrightarrow{\pa_1 \pb_2}$ in $\Traps_u'$ that contains
    $\pa$ on its leg, contains $\pb$ on its other leg.
\end{proof}

\subsubsection{Constructing the local spanner for nice %
   polygons}

\begin{theorem}
    \thmlab{k:gon}%
    Let $\Body$ be a $k$-nice convex polygon, $\PS$ be a set of $n$
    points in the plane, and let $\eps \in (0,1)$ be a
    parameter. Then, one can construct a $\CC$-local
    $(1+\eps)$-spanner of $\PS$.  The construction time is
    $\Of\pth{\bigl. ({k^4}/{\eps^6})n \log^2 n}$, and the resulting
    graph has $\Of\pth{\bigl. ( {k^4}/{\eps^6})n \log n}$ edges. In
    particular these bounds hold if $\Body$ is a $k$-regular polygon.
\end{theorem}
\begin{proof}
    Let $\epsA = \eps/\constD$, for $\constD$ sufficiently large
    constant.  We construct $\Triangles$, a family of triangles
    induced by a vertex of $\Body$, and an non-adjacent edge of
    $\Body$. This family has $\Of(k^2)$ triangles. Each such triangle
    is $\Omega(1/k)$-fat, and for each such triangle we construct the
    $(1+\epsA)$-spanner of \thmref{l:s:triangle} for \PS. Next, we
    cover $\CC$ by a set $\Traps$ of $k' = \Of(k^4/\epsA^3)$
    $\epsA$-narrow trapezoids using \lemref{narrow:traps:decomp}.

    We compute an $\epsA$-angular $(1/\epsA)$-\SSPD $\WS$
    decomposition of $\PS$ using \corref{S:S:P:D:angular} -- the total
    weight of the decomposition is
    $w = \Of\pth{ n \epsA^{-3} \log n }$. For each pair
    $\{\PX, \PY\} \in \WS$, and each trapezoid $\Trap \in \Traps$, we
    compute the $\Trap$-Delaunay triangulation of $\PX \cup \PY$.

    Let $\G$ denote the union of all these graphs. We claim that it is
    the desired spanner.  The construction time is
    \begin{equation*}
        \Of(( k^3 /\epsA) n \log n + k' w \log n)
        =%
        \Of\pth{
           \frac{k^3 }{\epsA} n \log n + \frac{k^4}{\epsA^3} \cdot
           \frac{n}{\epsA^{3}} \log n \cdot \log n
        }%
        =%
        \Of\pth{ \frac{k^4}{\epsA^6}n \log^2 n},%
    \end{equation*}
    and the resulting graph has
    $\Of\pth{\bigl. ( {k^4}/{\epsA^6})n \log n}$ edges.
    
    As for correctness, consider a homothet $\Body'$ of $\Body$ that
    contains two points $\pa, \pb \in \PS$. By \lemref{shrink:shrank},
    there is a homothet $\Body'' \subseteq \Body'$ of $\Body$ such
    that $\pa, \pb \in \partial \Body''$. There are two possibilities:

    \smallskip%
    \noindent
    \textbf{(A)} The point $\pa$ is on a vertex of $\Body''$ and $\pb$
    is on an edge. In this case, the vertex and the edge induce a fat
    triangle, that is a homothet of a triangle
    $\triangle \in \Triangles$. Since the graph $\G$ contains a
    $\triangle$-local $(1+\eps)$-spanner for $\PS$, it follows readily
    that $\G$ is a $(1+\eps)$-spanner for these points, and the path
    is strictly inside $\Body''$.

    \smallskip%
    \noindent%
    \textbf{(B)} %
    The points $\pa$ and $\pb$ are on two non-adjacent edges of
    $\Body''$. Then, there is an $\epsA$-narrow trapezoid $\Trap'$
    that has $\pa$ and $\pb$ on its two legs, and a homothet of
    $\Trap'$, denoted by $\Trap$, is in $\Traps$. There is a pair
    $\{\PX, \PY\} \in \WS$ that is $(1/\epsA)$-semi separated (and
    $\epsA$-angularly separated), such that $\pa \in \PX$ and
    $\pb \in \PY$.  By \lemref{good:jump:traps}, there are two points
    $\pa' \in \PX$ and $\pb' \in \PY$, such that $\pa'\pb'$ is an edge
    of the $\Trap$-Delaunay triangulation of $\PX \cup \PY$, and by
    construction this edge is in $\G$. We now use induction on the
    shortest paths from $\pa$ to $\pa'$ and from $\pb$ to $\pb'$ in
    $\G$.  By induction, and \lemref{good:jump:traps}, we have that
    \begin{equation*}
        \dGY{\pa}{\pb} %
        \leq%
        \dGY{\pa}{\pa'} + \dY{\pa'}{\pb'} + \dGY{\pb'}{\pb} 
        \leq             
        (1+\eps)\dY{\pa}{\pa'} + \dY{\pa'}{\pb'} + (1+\eps)
        \dY{\pb'}{\pb}            
        \leq%
        (1+\eps) \dY{\pa}{\pb},
    \end{equation*}
    which implies that the there is $(1+\eps)$-path from $\pa$ to
    $\pb$ inside $\Body'$.
\end{proof}

\begin{remark}
    \remlab{improved}%
    For axis-parallel squares \thmref{k:gon} implies a local spanner
    with $\Of\pth{\eps^{-6} n \log n}$ edges.  However, for this
    special case, the decomposition into narrow trapezoid can be
    skipped. In particular, in this case, the resulting spanner has
    $\Of( \eps^{-3} n \log n)$ edges. We do not provide the details
    here, as it is only a minor improvement over the above, and
    requires quite a bit of additional work -- essentially, one has to
    prove a version of \lemref{good:jump:traps} for squares.
\end{remark}

\section{Weak local spanners for axis-parallel rectangles}
\apndlab{rectangles}

\subsection{Quadrant separated pair decomposition}
\apndlab{qspd}

For two points $\pa = (p_1, \ldots, p_d)$ and
$\pb = (q_1, \ldots, q_d)$ in $\Re^d$, let $\pa \prec \pb$ denotes
that $\pb$ \emphi{dominates} $\pa$ coordinate-wise. That is
$p_i < q_i$, for all $i$. More generally, let $\pa <_i \pb$ denote
that $\pa_i < \pb_i$. For two point sets $\PSX, \PSY \subseteq \Re^d$,
we use $\PSX <_i \PSY$ to denote that
$\forall \px \in \PSX, \py \in \PSY \quad \px <_i \py$.  In particular
$\PSX$ and $\PSY$ are \emphw{$i$-coordinate separated} if
$\PSX <_i \PSY$ or $\PSY <_i \PSX$. A pair $\{ \PSX, \PSY\}$ is
\emphi{quadrant-separated}, if $\PSX$ and $\PSY$ are $i$-coordinate
separated, for $i=1,\ldots, d$.

A \emphi{quadrant-separated pair decomposition} of a point set
$\PS \subseteq \Re^d$, is a pair decomposition (see
\defref{pair:decomposition})
$\WS = \bigl\{ \{ \PSX_1, \PSY_1 \}, \ldots, \{ \PSX_s, \PSY_s \}
\bigr\}$ of $\PS$, such that $\{ \PSX_i, \PSY_i\}$ are
quadrant-separated for all $i$.

\begin{lemma}
    \lemlab{d:1}%
    Given a set $\PS$ of $n$ points in $\Re$, one can compute, in
    $\Of( n \log n)$ time, a \QSPD of $\PS$ with $\Of(n)$ pairs, and
    of total weight $\Of( n \log n)$.
\end{lemma}
\begin{proof}
    If $\PS$ is a singleton then there is nothing to do. If
    $\PS = \{ \pa, \pb \}$, then the decomposition is the pair formed
    by the two singleton points.
	
    Otherwise, let $x$ be the median of $\PS$, such that
    $\PS_{\leq x} = \Set{ \pa \in \PS }{\pa \leq x}$ contains exactly
    $\ceil{n/2}$ points, and $\PS_{> x} = \PS \setminus \PS_{\leq x}$
    contains $\floor{n/2}$ points. Construct the pair
    $\Pair = \{ \PS_{\leq x}, \PS_{> x} \}$, and compute recursively a
    \QSPD{}s $\QS_{\leq x}$ and $\QS_{> x}$ for $\PS_{\leq x}$ and
    $\PS_{> x}$, respectively. The desired \QSPD is
    \begin{math}
	\QS_{\leq x} \cup \QS_{> x} \cup \{ \Pair \}.
    \end{math}
    The bounds on the size and weight of the desired \QSPD are
    immediate.
\end{proof}

\begin{lemma}
    Given a set $\PS$ of $n$ points in $\Re^d$, one can compute, in
    $\Of( n \log^d n)$ time, a \QSPD of $\PS$ with
    $\Of(n \log^{d-1} n)$ pairs, and of total weight
    $\Of( n \log^d n)$.
\end{lemma}
\begin{proof}
    The construction algorithm is recursive on the dimensions, using
    the algorithm of \lemref{d:1} in one dimension.
	
    The algorithm computes a value $\alpha_d$ that partitions the
    values of the points' $d$\th coordinates roughly equally (and is
    distinct from all of them), and let $h$ be a hyperplane parallel
    to the first $d-1$ coordinate axes, and having value $\alpha_d$ in
    the $d$\th coordinate.
	
    Let $\PS_{\uparrow}$ and $\PS_{\downarrow}$ be the subset of
    points of $\PS$ that are above and below $h$, respectively. The
    algorithm recursively computes \QSPD{}s $\QS_\uparrow$ and
    $\QSdown$ for $\PSup$ and $\PS_{\downarrow}$, respectively.  Next,
    the algorithm projects the points of $\PS$ on $h$, let $\PS'$ be
    the resulting $d-1$ dimensional point set (after we ignore the
    $d$\th coordinate), and recursively computes a \QSPD $\QS'$ for
    $\PS'$.

    For a point set $\PSX' \subseteq \PS'$, let $\liftX{\PSX'}$ be the
    subset of points of $\PS$ whose projection on $h$ is $\PSX'$.  The
    algorithm now computes the set of pairs
    \begin{equation*}
	\widehat{\QS}%
	=%
	\Set{\Bigl.
           \{ \liftX{ \PSX'} \cap \PSup , \liftX{ \PSY' } \cap \PSdown
           \}, \,\,
           \{ \liftX{ \PSX'} \cap \PSdown , \liftX{ \PSY' } \cap \PSup
           \}
	}%
	{ \{ \PSX', \PSY' \} \in \QS'}     .
    \end{equation*}
    The desired \QSPD is $\widehat{\QS} \cup \QSup \cup \QSdown$.
	
    To observe that this is indeed a \QSPD, observe that all the pairs
    in $\QSup, \QSdown$ are quadrant separated by induction. As for
    pairs in $\widehat{\QS}$, they are quadrant separated in the first
    $d-1$ coordinates by induction on the dimension, and separated in
    the $d$ coordinate since one side of the pair comes from $\PSup$,
    and the other side from $\PSdown$.
	
    As for coverage, consider any pair of points $\pa, \pb \in \PS$,
    and observe that the claim holds by induction if they are both in
    $\PSup$ or $\PSdown$. As such, assume that $\pa \in \PSup$ and
    $\pb \ni \PSdown$. But then there is a pair
    $\{\PSX', \PSY'\} \in \QS'$ that separates the two projected
    points in $h$, and clearly one of the two lifted pairs that
    corresponds to this pair quadrant-separates $\pa$ and $\pb$ as
    desired.
	
    The number pairs in the decomposition is
    \begin{math}
	N(n,d) = 2N(n,d-1) + 2N\pth{ n/2, d }
    \end{math}
    with $N(n,1) = \Of(n)$. The solution to this recurrence is
    $N(n,d) = \Of( n \log^{d-1} n)$.  The total weight of the
    decomposition is
    \begin{math}
	W(n,d) = 2W(n,d-1) + 2W\pth{ n/2, d }
    \end{math}
    with $W(n,1) = \Of(n \log n)$. The solution to this recurrence is
    $W(n,d) = \Of( n \log^{d} n)$. Clearly, this also bounds the
    construction time.
\end{proof}

\subsection{Weak local spanner for axis parallel rectangles}
\apndlab{w:l:s:rect}

For a parameter $\delta \in (0,1)$, and an interval $I = [b,c]$, let
$(1-\delta)I = [t - (1-\delta)r, t+ (1-\delta)r]$, where
$t = (b+c)/2$, and $r = (c-b)/2$, be the shrinking of $I$ by a factor
of $1-\delta$.

Let $\Rects$ be the set of all axis parallel rectangles in the
plane. For a rectangle $\rect \in \Rects$, with $\rect = I \times J$,
let $(1-\delta)\rect = (1-\delta)I \times (1-\delta)J$ denote the
rectangle resulting from shrinking $\rect$ by a factor of $1-\delta$.

\begin{defn}
    Given a set $\PS$ of $n$ points in the plane, and parameters
    $\eps, \delta \in (0,1)$, a graph $\G$ is a
    \emphw{$(1-\delta)$-local $(1+\eps)$-spanner} for rectangles, if
    for any axis-parallel rectangle $\rect$, we have that
    $\restrictY{\G}{\rect}$ is a $(1+\eps)$-spanner for all the points
    in $\bigl((1-\delta)\rect\bigr) \cap \PS$.
\end{defn}

Observe that rectangles in $\Rects$ might be quite ``skinny'', so the
previous notion of shrinkage used before is not useful in this case.

\subsubsection{Construction for a single quadrant separated pair}
\seclab{pair:edges}

\begin{figure}[t]
    \centering%
    \includegraphics{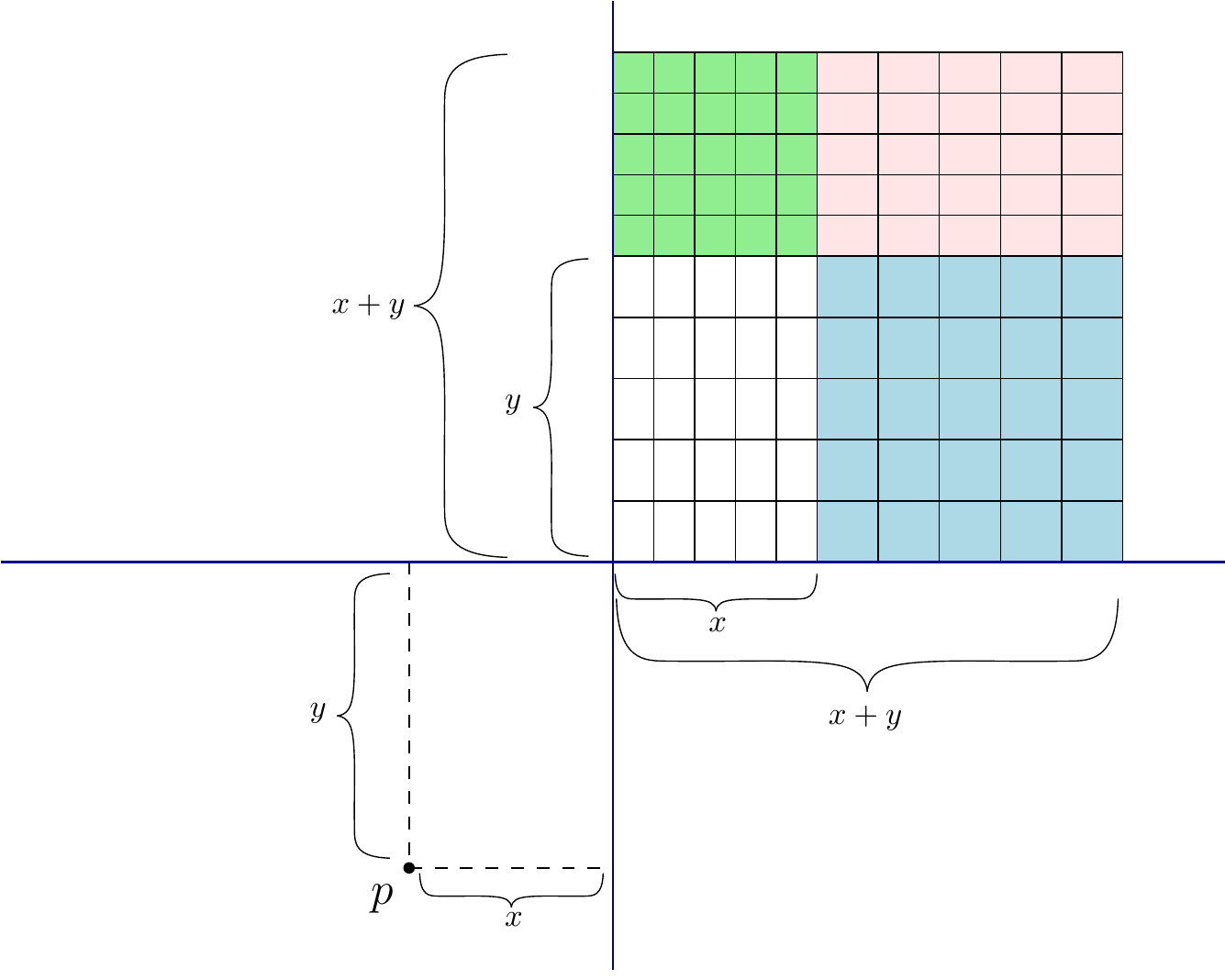}%
	\caption{The construction of the grid $\grid(\pa,\Pair)$ for a
		point $\pa=(-x,-y)$ and a pair $\Pair$.
             }
             \figlab{grid}%
         \end{figure}

         Consider a pair $\Pair = \{\PSX, \PSY \}$ in a \QSPD of
         $\PS$. The set $\PSX$ is quadrant-separated from $\PSY$. That
         is, there is a point $\cen_\Pair$, such that $\PSX$ and
         $\PSY$ are contained in two opposing quadrants in the
         partition of the plane formed by the vertical and horizontal
         line through $\cen_\Pair$.

         For simplicity of exposition, assume that
         $\cen_\Pair = (0,0)$, and $\PSX \prec (0,0 ) \prec
         \PSY$. That is, the points of $\PSX$ are in the negative
         quadrant, and the points of $\PSY$ are in the positive
         quadrant.

         We construct a non-uniform grid $\grid(\pa, \Pair)$ in the
         square $[0,x+y]^2$.  To this end, we first partition it into
         four subrectangles
         \begin{equation*}
             \begin{array}{l|l}
               \rectA_\nwarrow = [0,x] \times [y,x+y]
               &
                 \rectA_\nearrow = [x,x+y] \times [y,x+y]\Bigr.\\
               \hline
               \rectA_{\swarrow} = [0,x]\times [0,y]
               &
                 \rectA_\searrow = [x,x+y] \times [0,y].\Bigr.\\
             \end{array}
         \end{equation*}

         Let $\gConst \geq 4/\eps + 4/\delta$ be an integer number.
         We partition each of these rectangles into a
         $\gConst \times \gConst$ grid, where each cell is a copy of
         the rectangle scaled by a factor of $1/\gConst$.  See
         \figref{grid}. This grid has $\Of(\gConst^2)$ cells. For a
         cell $\cell$ in this grid, let $\PSY \cap \cell$ be the
         points of $\PSY$ contained in it. We connect $\pa$ to the
         left-most and bottom-most points in $\PSY \cap \cell$. This
         process generates two edges in the constructed graph for each
         grid cell (that contains at least two points), and
         $\Of( \gConst^2)$ edges overall.

         The algorithm repeats this construction for all the points
         $\pa \in \PSX$, and does the symmetric construction for all
         the points of $\PSY$.

         \subsubsection{The construction algorithm}

         The algorithm computes a \QSPD $\WS$ of $\PS$. For each pair
         $\Pair \in \WS$, the algorithm generates edges for $\Pair$
         using the algorithm of \secref{pair:edges} and adds them to
         the generated spanner $\G$.

\begin{figure}
    \phantom{}\hfill%
    \includegraphics[page=1]{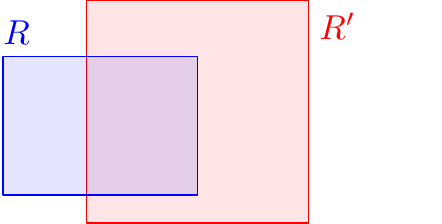}%
    \hfill%
    \includegraphics[page=2]{figs/expand}%
    \hfill\phantom{}%
    \caption{Left: The two rectangles $\rect$, $\rect'$. Right: In
       green $\xSlabX{\rect} \cap \rect'$, the restriction of the slab
       $\xSlabX{\rect}$ to the rectangle $\rect'$.}
    \figlab{x:expand}
\end{figure}

\subsubsection{Correctness}

For a rectangle $\rect$, let
$\xSlabX{\rect} = \Set{(x,y)\in \Re^2}{\exists (x',y) \in \rect}$ be
its expansion into a horizontal slab. Restricted to a rectangle
$\rect'$, the resulting set is $\xSlabX{\rect} \cap \rect'$, depicted
in \figref{x:expand}.  Similarly, we denote
\begin{equation*}
    \ySlabX{\rect} = \Set{(x,y)\in \Re^2}{\exists (x,y') \in \rect}.
\end{equation*}

\begin{figure}[h]
    \includegraphics[page=2]{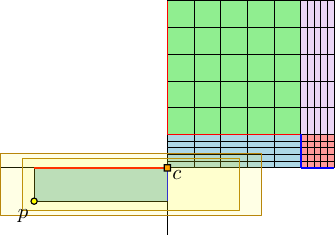}%
    \hfill%
    \includegraphics[page=3]{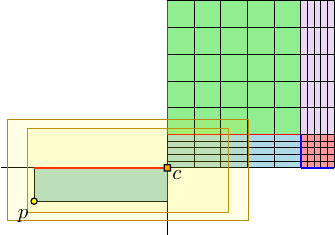}%
    \hfill%
    \includegraphics[page=4]{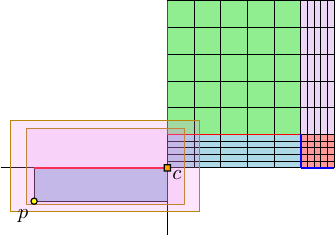}
    \caption{An illustration of $\grid(\pa, \Pair)$ with three
       rectangles and their shrunken version.}
    \figlab{grid:2}
\end{figure}
\begin{lemma}
    \lemlab{cases}%
    Assume that $\gConst \geq \ceil{20/\eps + 20/\delta}$.  Consider a
    pair $\Pair = \{\PSX, \PSY \}$ in the above construction, and a
    point $\pa =(-x,-y) \in \PSX$ with its associated grid
    $\grid = \grid(\pa, \Pair)$. Consider any axis parallel rectangle
    $\rect$, such that $\pa \in (1-\delta)\rect = I\times J$, and
    $(1-\delta)\rect$ intersects a cell $\cell \in \grid$. We have
    that:
    \begin{compactenumI}
        \smallskip%
        \item \itemlab{i:c} If $\cell \subseteq (1-\delta)\rect$ then
        $(1-\delta)^{-1} \cell \subseteq \rect$.
		
        \item \itemlab{small:diam}
        $\diameterX{\cell} \leq (\eps/4) \dsY{\pa}{\cell}$.

        \item \itemlab{x:bigger} If $x \geq y$ and
        $\cell \subseteq \rect_\swarrow \cup \rect_\searrow$ then
        $(1-\delta)^{-1} \cell \subseteq \rect$.
		
        \item \itemlab{y:bigger} If $x \leq y$ and
        $\cell \subseteq \rect_\swarrow \cup \rect_\nwarrow$ then
        $(1-\delta)^{-1} \cell \subseteq \rect$.
		
        \smallskip%
        \item \itemlab{n:w} If $x \geq y$ and
        $\cell \subseteq \rect_\nwarrow$, then
        $(1-\delta)^{-1} \bigl(\xSlabX{ (1-\delta)\rect} \cap \cell
        \bigr) \subseteq \rect$.
		
        \smallskip%
        \item \itemlab{s:e} If $x \leq y$ and
        $\cell \subseteq \rect_\searrow$, then
        $(1-\delta)^{-1} \Bigl(\ySlabX{\bigl( (1-\delta)\rect \bigr)}
        \cap \cell \Bigr) \subseteq \rect$.
    \end{compactenumI}
\end{lemma}
\begin{proof}
    \itemref{i:c} is immediate, \itemref{y:bigger} and \itemref{s:e}
    follows by symmetry from \itemref{x:bigger} and \itemref{n:w},
    respectively.
	
    \smallskip%
    \noindent%
    \itemref{small:diam} We have that
    \begin{math}
	\diameterX{\cell}%
	\leq%
	(x+y) /\gConst%
	=%
	\|\pa\|_1 /\gConst \leq %
	(\eps/4) \dsY{\pa}{\cell}.
    \end{math}
	
    \smallskip%
    \noindent%
    \itemref{x:bigger} The width, denoted $\widthX{\cdot}$, of
    $(1-\delta)\rect$ is at least $x$, as it contains both $\pa$ and
    the origin. As such,
    \begin{equation*}
	\bigl( \widthX{\rect} - \widthX{(1-\delta)\rect} \bigr)/2 \geq
	2(x /\tau) \geq 2 \widthX{\cell}.
    \end{equation*}
    That is, the width of the ``expanded'' rectangle $\rect$ is enough
    to cover $\cell$, and a grid cell adjacent to it to the right.
	
    A similar argument about the height shows that $\rect$ covers the
    region immediately above $\cell$ -- in particular, the vertical
    distance from $\cell$ to the top boundary of $\rect$ is at least
    the height of $\cell$. This implies that the expanded cell
    $(1-\delta)^{-1}\cell$ is contained in $\rect$, as claimed, as
    $\delta < 1/2$.

    \smallskip%
    \noindent%
    \itemref{n:w} We decompose the claim to the two dimensions of the
    region. Let
    $\rectA = \bigl(\xSlabX{ (1-\delta)\rect} \cap \cell
    \bigr)$. Observe that containment in the $x$-axis follows by
    arguing as in \itemref{x:bigger}. As for the $y$-interval of $B$,
    observe that it is contained in the $y$-interval of
    $(1-\delta)\rect$, which implies that when expanded by
    $(1-\delta)^{-1}$, it would be contained in the $y$-interval of
    $\rect$. Combining the two implies the result.
\end{proof}

\begin{lemma}
    \lemlab{local:spanner}%
    For any axis-parallel rectangle $\rect$, and any two points
    $\pa, \pb \in (1-\delta)\rect \cap \PS$, there exists a
    $(1+\eps)$-path between $\pa$ and $\pb$ in $\G$.
\end{lemma}
\begin{proof}
    The proof is by induction over the size of $\rect$ (i.e. area,
    width, or height). Let $\Pair = \{ \PSX, \PSY \} \in \WS$ be the
    pair in the \QSPD that separates $\pa$ and $\pb$, let $\cen$ be
    the separation point of the pair, and assume for the simplicity of
    exposition that $\pa \in \PSX$, $\PSX \prec \cen \prec \PSY$, and
    $\cen = (0,0)$. Furthermore, assume that
    $\|\pa\|_1 \geq \|\pb\|_1$.
	
    Let $\pa = (-x,-y)$, and let $\cell$ be the grid cell of
    $\grid(\pa, \Pair)$ that contains $\pb$. If
    $\cell \subseteq (1-\delta)\rect$, then
    $(1-\delta)^{-1}\cell \subseteq \rect$ by \lemref{cases}
    \itemref{i:c}. As such, let $\pc$ be the leftmost point in
    $\cell \cap \PS$. Both $\pb, \pc \in (1-\delta)^{-1}\cell$, and by
    induction, there is an $(1+\eps)$-path $\pi$ between them in $\G$
    (note that the induction applies to the two points, and the
    ``expanded'' rectangle $(1-\delta)^{-1}\cell$). Since $\pa \pc$ is
    an edge of $\G$, prefixing $\pi$ by this edge results in an
    $(1+\eps)$-path, as $\dY{\pb}{\pc} \leq (\eps/4) \dY{\pa}{\pb}$,
    by \lemref{cases} \itemref{small:diam} (verifying this requires
    some standard calculations which we omit).
	
    Otherwise, one need to apply the same argument using the
    appropriate case of \lemref{cases}.  So assume that $x \geq y$
    (the case that $y \geq x$ is handled symmetrically). If
    $\cell \subseteq \rect_\swarrow \cup \rect_\searrow$, then
    \itemref{x:bigger} implies that
    $(1-\delta)^{-1} \cell \subseteq \rect$. Which implies that
    induction applies, and the claim holds.
	
    The remaining case is that $x \geq y$ and
    $\cell \subseteq \rect_\nwarrow$.  Let
    $\rectB= \xSlabX{ (1-\delta)\rect} \cap \cell$.  By \itemref{n:w},
    we have $(1-\delta)^{-1}\rectB \subseteq \rect$. Namely,
    $\pb \in (1-\delta)\rect \cap \cell \subseteq \rectB$, and let
    $\pc$ be the lowest point in $\cell \cap \PS$. By construction
    $\pa \pc \in \EGX{\G}$, $\pb, \pc \in \rectB$,
    $(1-\delta)^{-1} \rectB \subseteq \rect$. As such, we can apply
    induction to $\pb$, $\pc$, and $(1-\delta)^{-1} \rectB$, and
    conclude that $\dGZ{\G}{\pb}{\pc} \leq (1+\eps) \dY{\pb}{\pc}$.
    Plugging this into the regular machinery implies the claim.
\end{proof}

\begin{theorem}
    \thmlab{a:l:s:rectangles}%
    Let $\PS$ be a set of $n$ points in the plane, and let
    $\eps, \delta \in (0,1)$ be parameters. The above algorithm
    constructs, in $\Of((1/\eps^2 + 1/\delta^2) n \log^2 n)$ time, a
    graph $\G$ with $\Of((1/\eps^2 + 1/\delta^2) n \log^2 n)$
    edges. The graph $\G$ is a $(1-\delta)$-local $(1+\eps)$-spanner
    for axis parallel rectangles. Formally, for any axis-parallel
    rectangle $\rect$, we have that $\rect \cap \PS$ is an
    $(1+\eps)$-spanner for all the points of
    $\bigl((1-\delta)\rect \bigr)\cap \PS$.
\end{theorem}
\begin{proof}
    Computing the \QSPD $\WS$ takes $\Of(n \log^2 n)$ time. For each
    pair $\{\PSX, \PSY\}$ in the decomposition with
    $m = |\PSX| + |\PSY|$ points, we need to compute the lowest and
    leftmost points in $(\PSX \cup \PSY) \cap \cell$, for each cell in
    the constructed grid. This can readily be done using orthogonal
    range trees in $\Of( \log^2 n)$ time per query (a somewhat faster
    query time should be possible by using that offline nature of the
    queries, etc). This yields the construction time. The size of the
    computed graph is
    $\Of(\WeightX{\WS} \gConst^2) = O\bigl((1/\delta^2 + 1/\eps^2) n
    \log^2 n\bigr)$.
	
    The desired local spanner property is provided by
    \lemref{local:spanner}.
\end{proof}

\BibTexMode{%
   \SoCGVer{%
      \bibliographystyle{plainurl}%
   }%
   \NotSoCGVer{%
      \bibliographystyle{alpha}%
   }%
   \bibliography{ft_spanner}%
}%
\BibLatexMode{\printbibliography}

\appendix

\section{Proof of \TPDF{\lemref{refine:d:w}}{{refine:d:w}}}
\apndlab{refine:d:w}

\RestatementOf{\lemref{refine:d:w}}%
{%
   \LemmaRefineDWBody{}%
}
   
\begin{figure}[ht]
    \centerline{\includegraphics{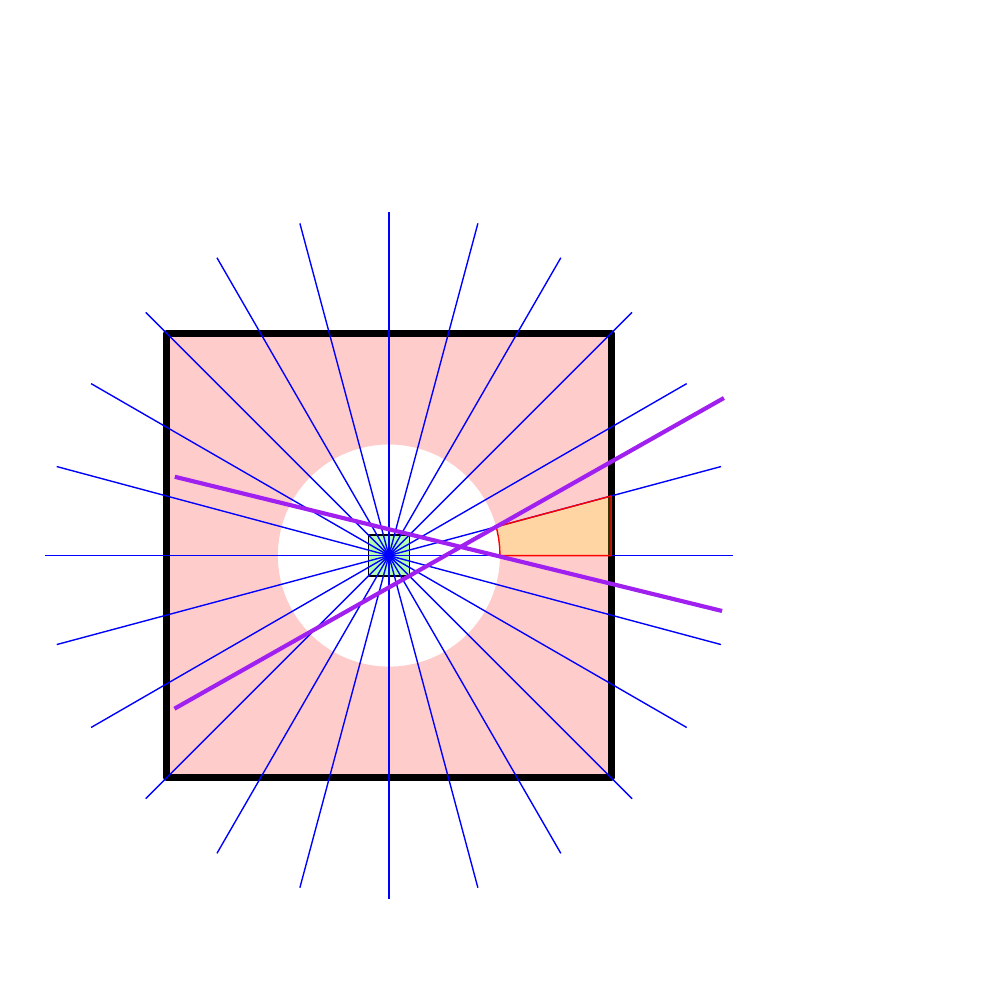}}
    \caption{An illustration of refining the pairs in a \SSPD into
       pairs contained in opposite parts of an
       $\eps$-double-wedge. $\PSX$ is contained in the green square
       $\square$, while $\PSY$ is contained in the red square, and the
       white gap between them is a result of the separation
       property. The set of cones with the apex at the center of
       $\square$ gives us the desired partition as demonstrated by the
       purple double-wedge. }
    \figlab{partition}
\end{figure}

\begin{proof}
    By using \lemref{chop:easy}, we can assume that $\WS$ is (say)
    $(10/\eps)$-separated.  Now, the algorithm scans the pairs of
    $\WS$. For each pair $\Pair = \{ \PSX, \PSY \} \in \WS$, assume
    that $\diameterX{\PSX} < \diameterX{\PSY}$. Let $\square$ be the
    smallest axis-parallel square containing $\PSX$, centered at point
    $\origin$.  Partition the plane around $\origin$, by drawing
    $\Of(1/\eps)$ lines intersecting $\origin$ with the angle between
    any two consecutive lines being at most (say) $\eps/4$, see
    \figref{partition}. This partitions the plane into a set of cones
    $\ConeSet$. For a cone $\cone \in \ConeSet$, we show that there
    exists an $\eps$-double-wedge that contains $\PSX$ in one side,
    and $\PSY \cap \cone$ in the other.

    To see that, take the double-wedge formed by the cross tangents
    between $\CHX{\PSX}$ and $\CHX{\PSY \cap \cone}$, where
    $\CHX{\PSX}$ denotes the convex-hull of $\PSX$. Assume w.l.o.g
    that $\square$ has side length 1, and let $\coneB$ be a cone of
    angle $\eps / 4$ with apex $\origin$, whose angular bisector is a
    horizontal ray in the positive direction of the $x$ axis. See
    figure \figref{double-wedge} for an illustration.

    \begin{figure}[h]
    \phantom{}\hfill%
    \includegraphics[page=2, width=0.48\linewidth]{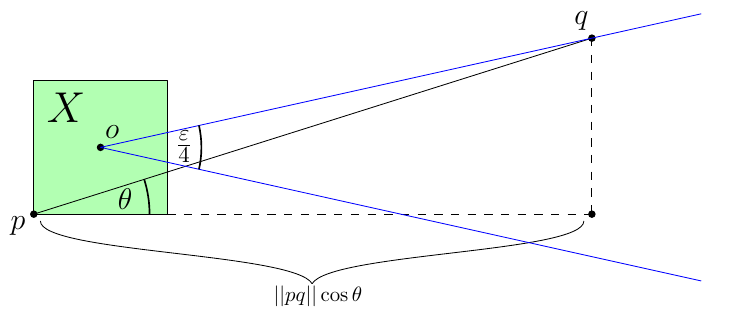}%
    \hfill%
    \includegraphics[page=3, width=0.48\linewidth]{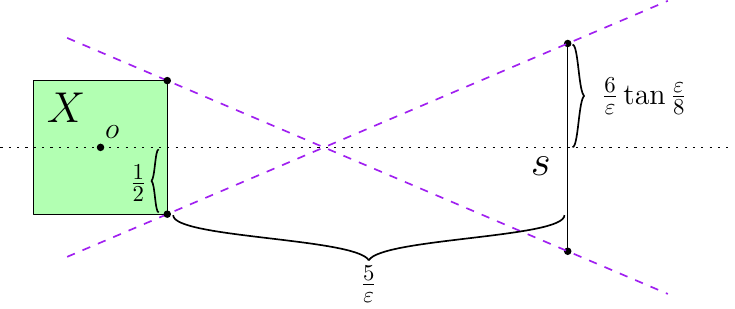}%
    \hfill%
    \phantom{}%
    \caption{An illustration of the proof for \lemref{refine:d:w}}
    \figlab{double-wedge}
\end{figure}

    We would like to find a vertical segment $\seg$ such that all
    points of $\PSY$ lie to its right, with one endpoint on the upper
    line of $\coneB$, and the other on the lower line of
    $\coneB$. Using the segments' height and distance from the right
    side of $\square$ we will be able to get a bound on the angle of
    the cross tangents. We first find a segment $\seg$ with all points
    of $\PSY$ to its right. A trivial bound on that distance is given
    by the segment from, say, the lower left corner of $\square$,
    denoted $\pa$, of length $10/\eps$ with its right endpoint on the
    upper line of $\coneB$, denote this point by $\pb$. We know that
    all points of $\PSY$ lie to the right of $\pb$ due to the
    $10/\eps$ separation property of the \SSPD. The segment $\pa\pb$
    creates an angle $\leq\pi/4$ with the $x$-axis (by the choice of
    the angle of $\coneB$).  We therefore get that the $x$-coordinate
    difference between $\square$ and $\pb$ is at most
    $10/\eps\cdot \cos\frac{\pi}{4}-1\leq 7/\eps-1\leq 6/\eps$. So let
    $\seg'$ be a vertical segment between the upper and lower rays of
    $\coneB$, with $x$-coordinate distance of $6/\eps-\frac{1}{2}$
    from $\square$ (in order to make calculations easier). We get that
    $s'$ is of length $2\cdot \frac{6}{\eps}\tan
    \frac{\eps}{8}$. Finally, we take $\seg$ to be a vertical segment
    of length $\frac{12}{\eps}\tan \frac{\eps}{8}$, with its center on
    the $x$-axis at a distance of $5/\eps+\frac{1}{2}$ away from
    $\origin$. The angle of the $x$-axis and the segment between the
    lower end of the right side of $\square$ and the upper end of
    $\seg$ is now given by:

    \begin{equation*}
 	\arctan
        \pth{\frac{\frac{6}{\eps}
              \tan\frac{\eps}{8}+\frac{1}{2}}{\frac{5}{\eps}}
        }%
        =%
        \arctan \pth{ \frac{6}{5}\tan\frac{\eps}{8}+\frac{\eps}{10} }%
        \leq%
        \eps%
    \end{equation*}
\end{proof}

\end{document}